\documentclass[english, 11pt]{article}
\usepackage{subfiles}
\usepackage{caption}
\usepackage{subcaption}
\usepackage{multirow}
\usepackage{diagbox}
\usepackage{graphicx}
\usepackage{booktabs}
\usepackage[colorlinks=true, citecolor=blue]{hyperref}
\usepackage{float}
\usepackage{amsthm}
\usepackage{amsmath,amsfonts,amssymb,stackengine,graphicx}
\usepackage{cleveref}

\usepackage[round, square, numbers]{natbib}
\newtheorem{theorem}{Theorem}[section]
\newtheorem*{theorem-non}{Theorem}

\newtheorem{ass}[theorem]{Assumption}
\newtheorem{lemma}[theorem]{Lemma}

\usepackage[algo2e,ruled,linesnumbered]{algorithm2e} %
\SetKwFunction{stop}{Stop}
\SetKwFunction{Range}{range}
\SetKw{KwTo}{in}
\SetKw{KwCon}{converged}

\newcommand{\simiid}{\stackrel{iid}{\sim}}
\newcommand{\nonstop}{$\phi_t$ \KwCon }
\newcommand{\RouLoop}{$k$ \KwTo $\{1,2,\cdots,R\}$}

\newcommand{\BatLoop}{$i$ \KwTo $\{1,2,\cdots,B\}$}
\newcommand{\Addloop}{$\ell$ \KwTo $\{{\underline{m}}+1,{\underline{m}}+2,\cdots,\ell_i$\}}

\usepackage[normalem]{ulem}
\newcommand{\mbe}{\mathbb{E}}
\newcommand{\mco}{\mathcal{O}}
\newcommand{\mlmc}{\mathrm{mlmc}}

\allowdisplaybreaks[3]

\usepackage[a4paper]{geometry}
\title{Leveraging Nested MLMC for Sequential Neural Posterior Estimation with Intractable Likelihoods}

\author{Xiliang Yang$^1$, Yifei Xiong$^2$,  Zhijian He$^1$\thanks{Corresponding author: hezhijian@scut.edu.cn}}
\date{%
    $^1$School of Mathematics, South China University of Technology\\%
    $^2$Department of Statistics, Purdue University\\
}

\begin{document}

\maketitle

\begin{abstract}
There is a growing interest in studying sequential neural posterior estimation (SNPE) techniques due to their advantages for simulation-based models with intractable likelihoods. The methods aim to learn the posterior from adaptively proposed simulations using neural network-based conditional density estimators. As an SNPE technique, the automatic posterior transformation (APT) method proposed by Greenberg et al. (2019) performs well and scales to high-dimensional data. However, the APT method requires computing the expectation of the logarithm of an intractable normalizing constant, i.e., a nested expectation. Although atomic proposals were used to render an analytical normalizing constant, it remains challenging to analyze the convergence of learning. In this paper, we reformulate APT as a nested estimation problem.  Building on this, we construct several multilevel Monte Carlo (MLMC) estimators for the loss function and its gradients to accommodate different  scenarios, including
two unbiased estimators, and a biased estimator that trades a small bias for reduced variance and controlled runtime and memory usage. 
We also provide convergence results of
stochastic gradient descent to quantify the interaction of  the bias and variance of the gradient estimator. 
Numerical experiments for approximating complex posteriors with multimodality in moderate dimensions are provided to examine the effectiveness of the proposed methods. 
\end{abstract}

\section{Introduction} \label{section:intro}
Simulator-based models are widely used across scientific disciplines, including neuroscience \cite{paninski2018neural}, physics \cite{brehmer2020mining, gonccalves2020training}, biology 
\cite{hashemi2023amortized,li2023biological, olutola2023systems}, and inverse graphics \cite{romaszko2017vision}. These models serve as crucial tools for describing and understanding the data-generating process. However, when applying traditional Bayesian inference to simulator-based models, challenges arise, including an intractable likelihood function $p(x|\theta)$ and the computational expense associated with running the simulator.

To address these challenges, a series of likelihood-free Bayesian computation (LFBC) methods have been developed. These methods include approximate Bayesian computation (ABC) \cite{beaumont2002approximate, marin2012approximate}, synthetic likelihoods (SL) \cite{price2018bayesian,wood2010statistical}, Bayesian optimization \cite{gutmann2016bayesian}, likelihood-free inference by ratio estimation \cite{thomas2022likelihood}, and pseudo-marginal methods \cite{andrieu2010particle,andrieu2009pseudo}. A comprehensive summary and review of these methods can be found in \cite{cranmer2020frontier}, and they have all been benchmarked in \cite{hermans2021trust,lueckmann2021benchmarking}. 

Posterior density estimation approaches approximate the posterior of interest $p(\theta|x_o)$ with a family of density estimators $q_\phi(\theta)$, where $\phi$ is the parameter to be tuned. Optimization-based approaches are widely used in these methods, and the Kullback-Leibler (KL) divergence between $p(\theta|x)$ and $q_\phi(\theta)$, which measures the difference between two densities, is commonly chosen as the loss function. Variational Bayes (VB), as an effective optimization-based method for approximating the posterior distribution of a Bayesian problem, is widely used. In the likelihood-free context, \cite{tran2017variational} developed a new VB method with an intractable likelihood, while \cite{he2022unbiased} proposed an unbiased VB method based on nested MLMC. However, these methods tend to fail in cases where simulations are expensive, as the nested estimation in these methods requires additional simulation procedures.

For this reason, there has been a growing interest recently in employing neural networks to represent probability densities, particularly normalizing flows~\cite{kobyzev2020normalizing}. When the inference problem focuses solely on the observation $x_o$, the data efficiency can be improved by using the sequential training schemes in the sequential neural posterior estimation (SNPE) methods \cite{apt, lueckmann2017flexible, papamakarios2016fast}. In this approach, model parameters are drawn from a proposal distribution that is more informative about $x_o$ than the prior distribution. However, SNPE requires modifying the loss function relative to neural posterior estimation (NPE) to ensure that the neural network approximates the true posterior $p(\theta|x)$. Many approaches have been developed to address this issue.

Among SNPE methods, automatic posterior transformation (APT) requires computing the expectation of a logarithm. This logarithm is of an intractable normalizing factor, and crucially, the factor itself is also an expectation, resulting in a nested expectation overall. Greenberg et al.~\citep{apt} then used atomic proposals 
to render an analytical normalizing constant. However, it remains challenging to perform convergence analysis using the existing analysis techniques \cite{bottou2018optimization}. To address this limitation, we reformulate APT as a nested estimation problem. We then develop several multilevel Monte Carlo (MLMC)~\cite{giles2015multilevel} nested estimators for the loss function and its gradients, which enjoy lower computational complexity than the single-level nested estimator. In stochastic gradient training, the usual unbiased MLMC estimator~\cite{rhee2015unbiased} may suffer from high gradient variance and large GPU memory footprints. Motivated by~\cite{hu2021bias}, we study two MLMC variants: (i) an unbiased scheme that reduces variance at additional computational cost, and (ii) a truncated scheme that introduces a small, controlled bias to reduce variance while adhering to compute and memory budgets. We establish several theorems that investigate the order of bias, variance, and the average cost of the MLMC estimators. Finally, we establish convergence guarantees for stochastic gradient descent (SGD) and quantify the effects of the bias and variance of the gradient estimator.

The remainder of this paper is organized as follows. In \Cref{section:biasedAPT}, we propose the nested APT method and analyze its computational complexity. In \Cref{section:UBnestedAPT}, we introduce the basics of MLMC methods, including their formulation and theoretical analysis of the random sequence utilized in MLMC methods. Equipped with these tools, we develop both unbiased MLMC methods and truncated MLMC methods and analyze the order of variance for their losses and gradients, as well as the average cost. In \Cref{section:conver_sgd}, we provide convergence analysis for nested APT and truncated MLMC methods in the case of SGD. In~\Cref{sec:UBAPT_exp}, we conduct a series of numerical experiments on benchmark tasks. Finally, we conclude this paper with some remarks in \Cref{section:discuss}. Lengthy proofs and additional results are deferred to the appendices.

\section{Nested APT} \label{section:biasedAPT}
\subsection{Problem formulation}

Assume that the prior distribution has density $p(\theta)$ with respect to the Lebesgue measure, where $\theta \in \Theta$ denotes the model parameter of interest. We focus on the finite-dimensional parameter space $\Theta$. Given an observed sample $x_o$, our objective is to perform inference on the posterior $p(\theta|x_o)\propto p(\theta) p(x_o|\theta)$. However, in many cases, the likelihood $p(x|\theta)$ either lacks a closed form or is difficult to evaluate directly. Instead, it can be expressed via a simulator so that we can sample $x$ from $p(x|\theta)$ for a fixed model parameter $\theta$.

Since our objective is to approximate the posterior of interest $p(\theta|x_o)$ using tractable density estimators, the KL divergence serves as a primary measure of the discrepancy between two densities. The KL divergence is defined as
\begin{equation*}
\mathcal{D}_{\mathrm{KL}}\left(p(\theta)\|q(\theta)\right)  = \int p(\theta)\log \frac{ p(\theta)}{q(\theta)}\mathrm{d}\theta,
\end{equation*}
which is nonnegative by Jensen's inequality and attains its minimum when $q(\theta)$ agrees with $p(\theta)$, making it suitable as a loss function. When the likelihood $p(x|\theta)$ is tractable, one can directly approximate $p(\theta|x_o)$ by minimizing the KL divergence between the target distribution and the proposed estimator $q(\theta)$ within a certain family of distributions. In the likelihood-free context, this can be viewed as a problem of conditional density estimation. Within this framework, a conditional density estimator $q_{F(x, \phi)}(\theta)$ based on a neural network \cite{papamakarios2016fast, apt} is utilized to approximate $p(\theta|x)$ over the admissible set of tuning parameters $\phi \in \Phi$, where $\Phi\subset \mathbb{R}^d$ is the space of the neural network parameters. To this end, we focus on minimizing the following \textit{average KL divergence} under the marginal distribution with density $p(x)=\int p(\theta) p(x|\theta) \mathrm{d} \theta$
\begin{align*}
&\quad \ \mathbb{E}_{ p(x)}\left[\mathcal{D}_{\mathrm{KL}}\left(p(\theta|x)\| q_{F(x,\phi)}(\theta)\right)\right] \notag\\
&= \iint p(x)p(\theta|x) \left(\log  p(\theta|x)-\log q_{F(x,\phi)}(\theta)\right) \mathrm{d}x\mathrm{d}\theta \notag\\
&=-\mathbb{E}_{p(\theta, x)}\left[\log q_{F(x,\phi)}(\theta)\right] + \iint  p(\theta,x) \log p(\theta|x) \mathrm{d}x\mathrm{d}\theta \notag\\
& := L_1(\phi) + \iint p(\theta,x) \log p(\theta|x) \mathrm{d}x\mathrm{d}\theta,
\end{align*}
where the term 
\begin{equation}\label{eq:npe}
    L_1(\phi):=-\mathbb{E}_{p(\theta, x)}\left[\log q_{F(x,\phi)}(\theta)\right],
\end{equation}
is used as the loss function. {Since $L_1(\phi)$ is intractable, we approximate it by the empirical estimator}
\begin{equation}\label{eq:npe_emp}
\hat{L}_1(\phi) = -\frac{1}{N}\sum_{i=1}^N \log q_{F(x_i,\phi)} (\theta_i),
\end{equation}
where the training data $\{(\theta_i, x_i)\}_{i=1}^N$ is sampled from the joint probability density $p(\theta, x)=p(\theta)p(x|\theta)$. After training, given the observation $x_o$, the posterior $p(\theta|x_o)$ can be approximated by $q_{F(x_o, \phi)}(\theta)$.

Since we aim to conduct conditional density estimation at $x_o$, it is essential to utilize a proposal $\tilde{p}(\theta)$ that is more informative regarding $x_o$ than the prior $p(\theta)$. After initializing $\tilde{p}(\theta)$ as $p(\theta)$, we then want the approximation of $p(\theta|x_o)$ to serve as a good candidate for the proposal in subsequent simulations. This conditional density estimation with an adaptively chosen proposal is called \textit{sequential neural posterior estimation} (SNPE). However, after replacing $p(\theta,x)$ with $\tilde{p}(\theta,x) = \tilde{p}(\theta)p(x|\theta)$ in  \eqref{eq:npe}, it is observed that $q_{F(x,\phi)}(\theta)$ approximates the \textit{proposal posterior}
\begin{equation}\label{eq:proposal_post}
    \tilde{p}(\theta|x):=p(\theta|x)\frac{\tilde{p}(\theta)p(x)}{p(\theta)\tilde{p}(x)}, 
\end{equation}
where $\tilde p(x)=\int \tilde p(\theta) p(x|\theta)\mathrm{d}\theta$. Hence, we adjust the loss function so that $q_{F(x,\phi)}(\theta)$ approximates the true posterior $p(\theta|x)$.

In APT \cite{apt}, the proposal distribution $\tilde{p}(\theta)$ is initialized as the prior distribution $p(\theta)$. Consequently, \eqref{eq:npe} can be directly used for the loss function. In the subsequent rounds, \cite{apt} proposed to replace $q_{F(x,\phi)}(\theta),\ p(\theta,x)$ with $\tilde{q}_{F(x,\phi)}(\theta)$, $\tilde{p}(\theta, x)$ in  \eqref{eq:npe} respectively. Explicitly, the loss function proposed in APT is
\begin{equation}\label{eq:apt_loss}
L_2(\phi):=\mathbb{E}_{\tilde{p}(\theta,x)}\left[-\log \tilde q_{F(x,\phi)}(\theta)\right],   
\end{equation}
where
\begin{equation}
\tilde{q}_{F(x,\phi)}(\theta) = q_{F(x,\phi)}(\theta)\frac{\tilde{p}(\theta)}{p(\theta)}\frac{1}{Z(x,\phi)},\quad
Z(x,\phi) = \int \frac{q_{F(x,\phi)}(\theta')}{p(\theta')} \tilde{p}(\theta')\mathrm{d}\theta',\label{eq:modi_density_est}
\end{equation}
$Z(x,\phi)$ here denotes the normalizing constant. Proposition 1 in \cite{papamakarios2016fast} shows that if $q_{F(x,\phi)}(\theta)$ is expressive enough that $\tilde{q}_{F(x,\phi^*)}(\theta) = \tilde{q}(\theta)$ for some parameter $\phi^*$, then $q_{F(x,\phi^*)}(\theta) = q(\theta|x)$. SNPE methods like APT iteratively refine the parameter $\phi$ and proposal $\tilde{p}(\theta)$ over a series of iterations, commonly called the `rounds' of training. In the $r$-th round, a distinct proposal $\tilde{p}_r(\theta)$ is used, leading to a different loss function.

Since the integral $Z(x,\phi)$ is usually intractable in practice, APT proposes the use of `atomic' proposals. APT with such proposals is known as atomic APT. Specifically, they assume that $\tilde{p}(\theta)$ is a discrete uniform distribution $\mathcal{U}\{\theta_1,\dots,\theta_M\}$, where $\theta_i$ are sampled from a distribution. The uniform setting of the proposal distribution enables the analytical computation of $Z(x,\phi)$. With this method, \eqref{eq:proposal_post} and \eqref{eq:modi_density_est} can be reformulated as
\begin{align}
    \tilde{p}(\theta|x)=\frac{p(\theta|x)/p(\theta)}{\sum_{i=1}^M p(\theta_i|x)/p(\theta_i)},\quad \tilde{q}_{F(x,\phi)}(\theta) =\frac{q_{F(x,\phi)}(\theta)/p(\theta)}{\sum_{i=1}^Mq_{F(x,\phi)}(\theta_i)/p(\theta_i)}.
\end{align}
Proposition 1 in \cite{apt} provides the consistency guarantees of atomic APT: given that $\theta_1,\dots,\theta_M$ are generated from a distribution that covers the target $p(\theta|x_o)$ support, atomic APT can recover the full posterior.

However, to the best of our knowledge, the use of atomic proposals makes it challenging to analyze their convergence behavior. Therefore, we are unable to explain the unexpected low performance in some tasks \cite{deistler2022truncated} with the existing convergence results \cite{bottou2018optimization}, both empirically and theoretically. As an alternative approach to estimating $Z(x,\phi)$, which enjoys a comprehensive theoretical framework and comparable performance, is studied in the next section.

\subsection{Nested estimation}\label{section:biasAPT}
For ease of presentation, we denote $$g_\phi(x,\theta):= \frac{q_{F(x,\phi)}(\theta)}{p(\theta)}.$$ We reformulate \eqref{eq:apt_loss} as 
\begin{align}
&\quad \ \mathbb{E}_{\tilde{p}(\theta,x)}\left[-\log \tilde q_{F(x,\phi)}(\theta)\right]\notag\\
&= -\mathbb{E}_{\tilde{p}(\theta,x)}\left[\log \frac{q_{F(x,\phi)}(\theta)}{p(\theta)}\right]+\mathbb{E}_{\tilde{p}(x)}[\log Z(x,\phi)]-\mathbb{E}_{\tilde{p}(\theta)}\left[\log \tilde{p}(\theta)\right]\notag\\
&= -\mathbb{E}_{\tilde{p}(\theta,x)}\left[\log g_\phi(x,\theta)\right] + \mathbb{E}_{\tilde{p}(x)}\left[\log \mathbb{E}_{\tilde{p}(\theta')}\left[g_\phi(x,\theta')\right]\right]-\mathbb{E}_{\tilde{p}(\theta)}\left[\log \tilde{p}(\theta)\right]\notag\\
&:=\mathcal{L}(\phi)-\mathbb{E}_{\tilde{p}(\theta)}\left[\log \tilde{p}(\theta)\right]\label{eq:nested_last},
\end{align}
where 
\begin{equation}\label{eq:refor_apt_loss}
    \mathcal{L}(\phi) := -\mathbb{E}_{\tilde{p}(\theta,x)}\left[\log g_\phi(x,\theta)\right] + \mathbb{E}_{\tilde{p}(x)}\left[\log \mathbb{E}_{\tilde{p}(\theta')}\left[g_\phi(x,\theta')\right]\right]
\end{equation} is selected as the loss function, and the last term $-\mathbb{E}_{\tilde{p}(\theta)}\left[\log \tilde{p}(\theta)\right]$ in \eqref{eq:nested_last} is dropped for it is independent of $\phi$. We employ stochastic gradient methods to optimize the loss function in this paper, and therefore $\nabla_\phi\mathcal{L}(\phi)$ is of interest. By interchanging expectation $\mathbb{E}[\cdot]$ and gradient operator $\nabla_{\phi}$, the gradient of \eqref{eq:refor_apt_loss} is given by
\begin{align}\label{eq:grad_refor_apt_loss}
    \nabla_\phi\mathcal{L}(\phi) := -\mathbb{E}_{\tilde{p}(\theta,x)}\left[\nabla_\phi\log g_\phi(x,\theta)\right] + \mathbb{E}_{\tilde{p}(x)}\left[\nabla_\phi\log \mathbb{E}_{\tilde{p}(\theta')}\left[g_\phi(x,\theta')\right]\right].
\end{align}
Notice that \eqref{eq:grad_refor_apt_loss} is similar to the gradient examined in the optimization procedure of Bayesian experimental design \cite{carlon2020nesterov,foster2019variational,huan2014gradient,kleinegesse2020bayesian}. 
Given any query point $(\theta,x)\sim \tilde{p}(\theta,x)$, we denote the corresponding queries of the loss and gradient as
\begin{align}
    \psi_\phi&:= -\log g_\phi(x,\theta)+\log Z(x,\phi)\label{eq:loss_query},\\
    \rho_\phi&:=-\nabla_\phi\log g_\phi(x,\theta)+\nabla_\phi \log Z(x,\phi)\label{eq:grad_query},
\end{align}
so that $\mathbb{E}_{\tilde{p}(\theta,x)}[\psi_\phi] = \mathcal{L}(\phi)$ and $\mathbb{E}_{\tilde{p}(\theta,x)}[\rho_\phi] = \nabla_\phi \mathcal{L}(\phi)$. 

Given the intractability of the normalizing constant $Z(x,\phi)=\mathbb{E}_{\tilde{p}(\theta')}\left[g_\phi(x,\theta')\right]$, it remains an obstacle to deriving an estimator for the loss function. A simple choice is leveraging its empirical estimator based on $M$ samples, which is given by
\begin{equation}\label{eq:emp_grad_loss}
 \hat{Z}_M(x,\phi) = \frac{1}{M}\sum_{j=1}^{M} g_\phi(x,\theta'_j),
\end{equation}
where $\theta'_1, \cdots, \theta'_M\sim\tilde{p}(\theta')$ independently, $M$ in the subscript denotes the number of the samples used for the estimation of $Z(x,\phi)$. We thus arrive at queries for nested estimators of \eqref{eq:loss_query} and \eqref{eq:grad_query} respectively
\begin{align}
    \psi_{\phi,M} &= -\log g_\phi(x,\theta)+\log \hat{Z}_M(x,\phi)=\log \frac{1}{M}\sum_{j=1}^M \frac{g_\phi(x,\theta'_j)}{g_\phi(x,\theta)},\label{eq:emp_loss_query}\\
    \rho_{\phi,M} &= -\nabla_\phi\log g_\phi(x,\theta)+\nabla_\phi \log \hat{Z}_M(x,\phi).\label{eq:emp_grad_query}
\end{align}

Nested estimators for the loss function \eqref{eq:refor_apt_loss} and its gradient \eqref{eq:grad_refor_apt_loss} are then given by the mean of $N$ independent identically distributed (iid) copies of their queries
\begin{align}
\hat{{\mathcal{L}}}^{\mathrm{Ne}}(\phi) &=  \frac{1}{N} \sum_{i=1}^N \psi_{\phi,M}^{(i)} = \frac{1}{N} \sum_{i=1}^N \log \frac{1}{M}\sum_{j=1}^M \frac{g_\phi(x_i,\theta'_{ij})}{g_\phi(x_i,\theta_i)}\label{eq:emp_bias_apt_loss},\\
\nabla_\phi\hat{\mathcal{L}}^{\text{Ne}}(\phi) &= \frac{1}{N}\sum_{i=1}^N \rho^{(i)}_{\phi,M}\label{eq:emp_bias_apt_grad},
\end{align}
where $\theta_{ij}'\simiid \tilde{p}(\theta')$ constitute the inner sample of size $M$ and $(\theta_i,x_i)\simiid \tilde{p}(\theta,x)$ constitute the outer sample of size $N$, and $\psi^{(i)}_{\phi,M},\ \rho^{(i)}_{\phi,M}$ are iid copies of $\psi_{\phi,M}$, $\rho_{\phi,M}$ respectively. Due to the nonlinearity of the logarithm, \eqref{eq:emp_bias_apt_loss} is a biased estimator. Using similar arguments to \cite{ryan2003estimating}, it is straightforward to show that the nested estimator $\hat{{\mathcal{L}}}^{\mathrm{Ne}}(\phi)$ has a variance of $\mco(1/N)$  and a bias of $\mco(1/M)$. 
For a given mean squared error (MSE) $\epsilon$, choosing $N=\mco(\epsilon^{-2})$ and $M=\mco(\epsilon^{-1})$ yields a computational complexity of $\mco(NM)=\mco(\epsilon^{-3})$. 

Plain nested estimators \eqref{eq:emp_loss_query}--\eqref{eq:emp_bias_apt_grad} are also known as single-level nested estimators compared to MLMC. Empirical results in \Cref{appendix:exp_nested_atomic_APT} show that  nested APT with the single-level nested estimator performs comparably to atomic APT. Casting the task as a nested estimation problem facilitates numerical analysis and further enables us to leverage MLMC techniques to enhance the single-level nested estimator.

\section{Multilevel nested APT} \label{section:UBnestedAPT}

\subsection{Basic idea of MLMC}\label{subsec:basic_mlmc}

We first introduce the idea of MLMC in a generic setting following \cite{giles2015multilevel}. Consider the problem of estimating $\mathbb{E}[P]$, where the sampling of $P$ is costly or even infeasible. Suppose that  $P$ can be approximated by tractable $P_\ell$ with increasing accuracy, but also increasing cost, as $\ell\to \infty$. As a result, the expectation $\mathbb{E}[P]$ can be estimated by the single-level Monte Carlo estimator, 
\begin{equation}\label{eq:singlelevel}
\hat{\mu}_{L,N}:=\frac 1N \sum_{i=1}^N P_L^{(i)},
\end{equation}
where $P_L^{(i)}$ are iid copies of $P_L$. Assume that $P_\ell$ has a bias $|\mathbb{E}[P_\ell - P]|=\mco(2^{-\alpha \ell})$ and an expected cost $\mco(2^{\gamma \ell})$, where $\alpha,\gamma>0$.  The estimator $\hat{\mu}_{L,N}$ then has an MSE of $\mco(2^{-2\alpha L})+\mco(1/N)$. To ensure an MSE of $\mco(\epsilon^2)$, one can take $L = \mco(\log(\epsilon^{-1/\alpha}))$ and $N = \mco(\epsilon^{-2})$, yielding an average computational cost of $\mco(N2^{\gamma L})=\mco(\epsilon^{-2-\gamma/\alpha})$. MLMC can greatly reduce the computational cost by performing most simulations with low accuracy (low level) at a low cost, with relatively few simulations being performed at high accuracy (high level) and a high cost. Specifically, using the simple identity
\begin{equation*}
\mbe[P_L] = \mbe[P_0] +\sum_{\ell=1}^L\mbe[P_\ell-P_{\ell-1}]=\sum_{\ell=0}^L\mbe[\Delta P_{\ell}],
\end{equation*}
where $\mbe[\Delta P_{\ell}]=\mbe[P_\ell-P_{\ell-1}]$ for $\ell\ge 1$ and $\mbe[\Delta P_{0}]=\mbe[P_0]$,
we can use the following unbiased estimator for $\mbe[P_L]$,
\begin{equation}\label{eq:mlmc}
\hat{\mu}_{\mlmc}:=\sum_{\ell=0}^L\left(\frac 1{N_{\ell}} \sum_{i=1}^{N_\ell}\Delta P_{\ell}^{(i)}\right),
\end{equation}
where $\Delta P_{\ell}^{(i)}$ are iid copies of $\Delta P_{\ell}$, the sample sizes $N_{\ell}$ decrease with level $\ell$.

Now suppose that $\Delta P_{\ell}$ has a variance of $\mco(2^{-r\ell})$ with an average cost of $\mco(2^{\gamma \ell})$. For a given MSE of $\mco(\epsilon^2)$, \cite{giles2015multilevel} showed that choosing proper finest level $L$ and  sample sizes $N_\ell$ render a computational complexity $C$ with bound
\begin{equation}\label{eq:mlmc_complexity}
\mathbb{E}[C] =
\begin{cases}
\mco(\epsilon^{-2}) & r > \gamma, \\
\mco(\epsilon^{-2} (\log \epsilon)^2) & r = \gamma, \\
\mco(\epsilon^{-2 - (\gamma - \beta)/\alpha}) & r < \gamma.
\end{cases}
\end{equation}
As stated before, it requires a computational complexity of $ \mco(\epsilon^{-2 - \gamma/\alpha})$ for the single-level Monte Carlo to achieve the same level of MSE. From the viewpoint of complexity, MLMC improves upon the single-level Monte Carlo for all cases of $r$ and achieves the optimal complexity of $\mco(\epsilon^{-2})$ when $r>\gamma$. 

Nested simulation combined with the MLMC method has been studied for other applications \cite{ giles:2018b, Giles2017DecisionmakingUU,goda2019multilevelMCest,he2022unbiased}. We next show how to use MLMC to estimate the gradient $\nabla_\phi \mathcal{L}(\phi)=\mathbb{E}\left[\rho_{\phi}\right]$, which is a nested expectation. To fit the MLMC setting, we can take $P = \rho_{\phi}$ and $P_\ell = \rho_{\phi,M_\ell}$, where $M_\ell = 2^{\ell}M_0 = 2M_{\ell-1}$. Note that the the base 2 is specifically chosen here solely for the convenience of subsequent analysis~\cite{goda2022unbiased}. Since $P_\ell$ converges to $P$ as $\ell\to\infty$, the core of MLMC is to construct coupling estimators $\Delta P_\ell$ with a variance decaying as $\mco(2^{-r\ell})$. To gain a large $r$, we use the  antithetic coupling construction as in \cite{bujok2015multilevel,Giles2017DecisionmakingUU,Michael2014antithetic,goda2019multilevelMCest}, which is better than the naive coupling $\Delta P_\ell=\rho_{\phi,M_\ell}-\rho_{\phi,M_{\ell-1}}$.
The fundamental concept of antithetic coupling involves the selection of two non-overlapping subsets, each of size $M_{\ell-1}$ from the $M_\ell$ inner samples $\{\theta_j'\}_{j=1}^{M_\ell}$ utilized for computing $\rho_{\phi, M_\ell}$ given in \eqref{eq:emp_grad_query}. This results in two independent realizations of $\rho_{\phi,M_{\ell-1}}$, identified as $\rho^{(a)}_{\phi,M_{\ell-1}}$ and $\rho^{(b)}_{\phi,M_{\ell-1}}$. Specifically,
\begin{equation*}
\rho^{(a)}_{\phi,M_{\ell-1}} =-\frac{\nabla_\phi g_\phi(x,\theta)}{g_\phi(x,\theta)}+\frac{\nabla g^{(a)}_{\phi,M_{\ell-1}}(x)}{g^{(a)}_{\phi,M_{\ell-1}}(x)},
\end{equation*}
where
\begin{equation*}
g^{(a)}_{\phi,M_{\ell-1}}(x) = \frac{1}{M_{\ell-1}}\sum_{j=1}^{M_{\ell-1}} g_\phi(x,\theta_j'),\ \nabla g^{(a)}_{\phi,M_{\ell-1}}(x) =  \frac{1}{M_{\ell-1}}\sum_{j=1}^{M_{\ell-1}}\nabla_\phi g_\phi(x,\theta_j').
\end{equation*}
The notations $\rho^{(b)}_{\phi,M_{\ell-1}}$, $\nabla g^{(b)}_{\phi,M_{\ell-1}}(x)$, and $g^{(b)}_{\phi,M_{\ell-1}}(x)$ are defined in a similar way by using  $\theta'_j$,  $j=M_{\ell-1}+1,\dots,M_{\ell}$ instead. 
Define 
\begin{equation}\label{eq:antit_delta_grad_query}
    \Delta\rho_{\phi,0}=\rho_{\phi, M_0},\quad \Delta\rho_{\phi,\ell}=\rho_{\phi,M_\ell}-\frac{1}{2}\left(\rho^{(a)}_{\phi,M_{\ell-1}} + \rho^{(b)}_{\phi,M_{\ell-1}}\right) \quad (\ell \geq 1). 
\end{equation}
It suffices to take $\Delta P_\ell= \Delta\rho_{\phi,\ell}$ since
\begin{equation*} %
\mathbb{E}\left[\Delta \rho_{\phi,0}\right] =  \mathbb{E}\left[\rho_{\phi,M_0}\right],\quad \mathbb{E}\left[\Delta \rho_{\phi,\ell}\right] = \mathbb{E}\left[\rho_{\phi,M_\ell}-\rho_{\phi,M_{\ell-1}}\right] \quad (\ell \geq 1).
\end{equation*}

In the case of estimating the loss function $\mathcal{L}(\phi)=\mathbb{E}[\psi_\phi]$, the antithetic construction of the coupling estimators $\{\Delta \psi_{\phi,\ell}\}_{\ell = 0}^{\infty}$ can be similarly given by
\begin{equation}\label{eq:antit_delta_loss_query}
    \Delta\psi_{\phi,0}=\psi_{\phi, M_0},\quad \Delta\psi_{\phi,\ell}=\psi_{\phi,M_\ell}-\frac{1}{2}\left(\psi^{(a)}_{\phi,M_{\ell-1}} + \psi^{(b)}_{\phi,M_{\ell-1}}\right) \quad (\ell \geq 1),
\end{equation}
where $\psi_{\phi,M_\ell}$ is defined in \eqref{eq:emp_loss_query}, and
$$ \psi_{\phi,M_{\ell-1}}^{(a)} =\log \frac{1}{M_{\ell-1}}\sum_{j=1}^{M_{\ell-1}} \frac{g_\phi(x,\theta'_j)}{g_\phi(x,\theta)},\quad \psi_{\phi,M_{\ell-1}}^{(b)} =\log \frac{1}{M_{\ell-1}}\sum_{j=M_{\ell-1}+1}^{M_{\ell}} \frac{g_\phi(x,\theta'_j)}{g_\phi(x,\theta)}.
$$

We have $\gamma=1$ for $\Delta\rho_{\phi,\ell}$ and $\Delta\psi_{\phi,\ell}$, it remains to examine the values of $\alpha$ and $r$. To this end, the next two theorems provide upper bounds for the first and second moments of $|\Delta \psi_{\phi,\ell}|$ and $\left\|\Delta \rho_{\phi,\ell}\right\|_2$ as $\ell\to\infty$. 
\begin{theorem}\label{theorem:var_delta_loss_query}
If there exist $p,q >2$ with $(p-2)(q-2)\geq 4$ such that for any $\phi\in\Phi$,
\begin{align*}
    \quad \mathbb{E}\left[\left(\frac{g_\phi(x,\theta)}{Z(x,\phi)}\right)^p\right]<\infty \quad and \quad \mathbb{E}\left[\left|\log\frac{g_\phi(x,\theta)}{Z(x,\phi)}\right|^q\right]<\infty,
\end{align*}
where $(\theta,x)\sim \tilde{p}(\theta)p(x|\theta)$, we have
\begin{align*}
    \mathbb{E}\left[|\Delta \psi_{\phi,\ell}|\right]= \mco(2^{-\ell}),\quad 
    \mathbb{E}\left[\Delta \psi_{\phi,\ell}^2\right]=\mco(2^{-r_1\ell}),
\end{align*}
where $r_1 = \min(p(q-2)/(2q),2)\in(1,2]$.
\end{theorem}

\begin{proof}
This proof follows an argument similar to Theorem 2 in \cite{goda2019multilevelMCest}, which particularly showed that $\mathbb{E}\left[|\Delta \psi_{\phi,\ell}|\right]= \mco(2^{-\alpha_1\ell})$ with $\alpha_1=\min(p(q-1)/(2q),1)$,
    Since $(p-2)(q-2)\ge 4$, we actually have $\alpha_1=1$.
\end{proof}

\begin{theorem}\label{theorem:var_delta_grad_query}
If there exists $s>2$ such that for any $\phi\in\Phi$,
    \begin{align*}
        \underset{x,\phi,\theta}{\sup} \|\nabla_\phi \log g_\phi(x,\theta)\|_{\infty} <\infty,\quad  \mathbb{E}\left[\left(\frac{g_\phi(x,\theta)}{Z(x,\phi)}\right)^s\right]<\infty,
    \end{align*}
    where the expectation is taken under the law $(\theta,x)\sim \tilde{p}(\theta)p(x|\theta)$, we then have  
    \begin{align*}
    \mathbb{E}\left[\left\|\Delta \rho_{\phi,\ell}\right\|_2\right] =\mco(2^{-\ell}),\quad \mathbb{E}\left[\|\Delta \rho_{\phi,\ell}\|_2^2\right] =\mco(2^{-r_2\ell}),
    \end{align*}
    where  $r_2 = \min (s,4)/2\in (1,2]$.
\end{theorem}

\Cref{theorem:var_delta_grad_query}  is based on the work of \cite{goda2022unbiased} and its proof is detailed in \Cref{appendix:pf_var_delta_grad_query}.
By \Cref{theorem:var_delta_loss_query} and \Cref{theorem:var_delta_grad_query} ,   $\Delta \psi_{\phi,\ell}$ and $\Delta\rho_{\phi,\ell}$ have a variance of $\mco(2^{-r \ell})$ with $r>\gamma=1$.  Using the fact that
$$
|\mbe[P_\ell-P]| = \left |\sum_{\ell'=\ell+1}^\infty \mbe[\Delta P_\ell]\right |\le \sum_{\ell'=\ell+1}^\infty \mbe[\left |\Delta P_\ell\right |],
$$
we have the bias rate $\alpha = 1$ for both cases.  According to \eqref{eq:mlmc_complexity}, the resulting MLMC estimator \eqref{eq:mlmc}  yields the optimal complexity $\mco(\epsilon^{-2})$. Note that the nested APT estimator has a complexity of $\mco(\epsilon^{-2-\gamma/\alpha})=\mco(\epsilon^{-3})$. Although the MLMC estimator \eqref{eq:mlmc} achieves lower complexity than the nested APT estimator, it is biased and depends on the prespecified error threshold $\epsilon$. Moreover, the sample sizes $N_\ell$ for each level and the finest level $L$ need to be tuned by the threshold $\epsilon$. For practical training problems, it may be difficult to choose the threshold $\epsilon$.
Since MLMC achieves the optimal regime $r>\gamma$, one can design unbiased MLMC estimators for eliminating the bias of the MLMC estimator \eqref{eq:mlmc}. The unbiased MLMC also admits the optimal complexity $\mco(\epsilon^{-2})$ without choosing the sample sizes $N_\ell$.

\subsection{Unbiased MLMC for APT}

In this subsection, we develop a series of unbiased estimators for the loss function $\mathcal{L}(\phi)$ and its gradient $\nabla_\phi \mathcal{L}(\phi)$ based on the MLMC approach.
\subsubsection{RU-MLMC}
Recall that both the loss function and its gradient can be represented as a sum of expectations over a sequence of random variables; this topic has been well studied. We introduce a non-negative integer-valued random variable $L$ that is independent of the two random variable sequences, with its probability mass function $\mathbb{P}(L = \ell) = w_\ell$. By the law of total expectation, we  have 
\begin{equation*}
\mathbb{E}\left[{\omega_L}^{-1}\Delta \rho_{\phi,L}\right]=\sum_{\ell=0}^\infty\mathbb{E}\left[{\omega_\ell}^{-1}\Delta \rho_{\phi,\ell}\right]P(L=\ell)=\nabla_\phi \mathcal{L}(\phi).
\end{equation*}
The equivalence between $\nabla_\phi \mathcal{L}(\phi)$ and the expectation of the query
\begin{align}
    V_{\mathrm{RU}} ={\omega_L}^{-1}{\Delta \rho_{\phi, L}}\label{eq:query_ru_grad},
\end{align}
leads to an unbiased MLMC estimator of the gradient $\nabla_\phi \mathcal{L}(\phi)$.

We similarly define $$U_{\mathrm{RU}} = {\omega_L}^{-1}{\Delta \psi_{\phi, L}},$$ for the estimation of the loss function $\mathcal{L}(\phi)$. This method is known as the randomized unbiased multilevel Monte Carlo (RU-MLMC) method \cite{rhee2015unbiased}. In this study, 
we take $L$ as a geometric distribution $\mathrm{Ge}(p)$ with $\omega_\ell = (1-p)^\ell p$ and $p = 1-2^{-\alpha}$~\cite{goda2022unbiased}. To ensure a finite variance and finite expected computational cost for $V_{\mathrm{RU}}, U_{\mathrm{RU}}$, it is required that $\alpha \in (1,\min \left(r_1,r_2\right))$, where $r_1$ and $r_2$ are from \Cref{theorem:var_delta_loss_query} and \Cref{theorem:var_delta_grad_query}, respectively. This is further examined with the following theorem. Since the gradient $V_{\mathrm{RU}}$ is a vector, the notation $\mathrm{Var}[V_{\mathrm{RU}}]$ used in the following applies to the variance of each component of $V_{\mathrm{RU}}$ for simplicity. Denote the cost of RU-MLMC by $\mathrm{Cost}_{\mathrm{RU}}$.

\begin{theorem}\label{theorem:ru_prop_collection}
Consider the settings of \Cref{theorem:var_delta_loss_query} and \Cref{theorem:var_delta_grad_query}. If $L\sim \mathrm{Ge}(p)$ with $p = 1-2^{-\alpha}$ and $1<\alpha < \min \left(r_1,r_2\right)$, then we have
    \begin{align*}
\mathrm{Var}\left[U_{\mathrm{RU}}\right]  &\leq \frac{A}{(1-2^{\alpha-r_1})(1-2^{-\alpha})},\\
\mathrm{Var}\left[V_{\mathrm{RU}} \right]&\leq \frac{B}{(1-2^{\alpha-r_2})(1-2^{-\alpha})},\\
        \mathrm{Cost}_{\mathrm{RU}} &\propto M_0\frac{2^\alpha-1}{2^\alpha-2},
    \end{align*}
    where the constants $A, B$ are independent of $\alpha$,  $r_1$ and $r_2$.
\end{theorem}

\begin{proof}
It is sufficient to bound the second moments of $U_{\mathrm{RU}}$ and $V_{\mathrm{RU}}$. First,  by taking $w_\ell = (1-p)^\ell p = (1-2^{-\alpha})2^{-\alpha\ell}$ and using \Cref{theorem:var_delta_loss_query}, we have
    \begin{align*}
\mathrm{Var}\left[U_{\mathrm{RU}}\right]&\le \mathbb{E}\left[\left(\frac{\Delta \psi_{\phi,L}}{w_L}\right)^2\right] = \sum_{\ell=0}^\infty \frac{\mathbb{E}\left[\Delta \psi_{\phi,\ell}^2\right]}{w_\ell}\le \frac{A}{(1-2^{-\alpha})}\sum_{\ell=0}^\infty2^{(\alpha-r_1)\ell}\\
&=\frac{A}{(1-2^{\alpha-r_1})(1-2^{-\alpha})}.
\end{align*}
Bounding $\mathrm{Var}[V_{\mathrm{RU}}]$ follows analogously from \Cref{theorem:var_delta_grad_query}. 

As for the average cost, we have the following
    \begin{align*}
       \mathrm{Cost}_{\mathrm{RU}} \propto \sum_{\ell=0}^\infty w_\ell M_\ell=M_0(1-2^{-\alpha})\sum_{\ell=0}^\infty 2^{(1-\alpha)\ell}=M_0\frac{2^\alpha-1}{2^\alpha-2}.
    \end{align*}
\end{proof}

\Cref{theorem:ru_prop_collection} indicates that a large $\alpha$ leads to less expected total computational burden, but larger variance for the estimators of loss and its gradient.

\subsubsection{GRR-MLMC}
RU-MLMC methods can at times suffer from excessive variance, which in turn degrades convergence behavior \cite{goda2022unbiased,he2022unbiased}. To address this, we employ alternative schemes that reduce variance while preserving unbiasedness. The Russian Roulette (RR) estimator \cite{lyne2015russian} is also employed to estimate the sum of an infinite series, wherein the evaluation of any term in the series only demands a finite amount of computation. This estimator relies on randomized truncation and assigns a higher weight to each term to accommodate the possibility of not computing them.
The gradient estimator of RR-MLMC is given by
\begin{equation}\label{eq:rr_query_grad}
    V_{\mathrm{RR}} := \sum_{j=0}^{L}\frac{\Delta \rho_{\phi,j}}{p_{j}},
\end{equation}
where $p_{j} =\mathbb{P}(L\geq j)=\sum_{\ell=j}^\infty w_\ell$ for $j\ge 0$.
If all $p_j> 0$, then $V_{\mathrm{RR}}$ is unbiased since
\begin{align*}
    \mathbb{E}[V_{\mathrm{RR}}] &= \sum_{\ell=0}^\infty \mathbb{P}(L=\ell)\sum_{j=0}^{\ell}\frac{\mathbb{E}[\Delta \rho_{\phi,j}]}{p_{j}}= \sum_{j=0}^\infty \frac{\mathbb{E}[\Delta \rho_{\phi,j}]}{p_{j}}\sum_{\ell\ge j}^\infty \mathbb{P}(L=\ell)\\
    &=\sum_{j=0}^\infty \mathbb{E}[\Delta \rho_{\phi,j}]=\nabla_\phi \mathcal{L}(\phi).
\end{align*}

Motivated by \cite{luo2019sumo}, to trade lower variance for higher cost based on the conventional RR-MLMC, one way is to ensure that the first ${\underline{m}}$ terms of the infinite series are always computed, where this ${\underline{m}}$ is called the $\textit{base level}$. In this case, the random index $L$ is set to have a lower bound ${\underline{m}}$, i.e., $\mathbb{P}(L\ge {\underline{m}})=1$, implying $p_j=1$ for all $j\le {\underline{m}}$. We call this modified estimator the generalized Russian roulette (GRR) estimator. When ${\underline{m}}=0$, this degenerates to the conventional RR-MLMC. The associated query for the GRR-MLMC estimator of the gradient is
\begin{align}
    V_{\mathrm{GRR}} := \rho_{\phi,M_{{\underline{m}}}}+\sum_{j={\underline{m}}+1}^{L}\frac{\Delta\rho_{\phi,j}}{p_{j}} \label{eq:query_grr_grad}.
\end{align}
In this case, the probability mass function of $L$ is chosen as 
$\mathbb{P}(L={\underline{m}})=1-\sum_{\ell>{\underline{m}}} w_\ell$ and $\mathbb{P}(L=\ell)=w_\ell$ for $\ell >{\underline{m}}$, where $w_\ell = (1-p)^\ell p$ and $p = 1-2^{-\alpha}$. We denote this distribution by $\mathrm{Ge}(p,{\underline{m}})$. 

Similarly, we define the query for the GRR estimator of the loss function:
\begin{align}
    U_{\mathrm{GRR}} := \psi_{\phi,M_{{\underline{m}}}}+\sum_{j={\underline{m}}+1}^{L}\frac{\Delta\psi_{\phi,j}}{p_{j}} \label{eq:query_grr_loss}.
\end{align}

We next provide an upper bound for the variance of the loss $U_{\mathrm{GRR}}$. It is sufficient to bound $\mathrm{Var}\left[\psi_{\phi,M_{{\underline{m}}}}\right]$ and $\mathrm{Var}\left[\sum_{j={\underline{m}}+1}^{L}\frac{\Delta\psi_{\phi,j}}{p_{j}}\right]$. Directly applying \Cref{theorem:var_delta_loss_query} and the law of total expectation, we have 
\begin{align*}
\mathrm{Var}\left[\sum_{j={\underline{m}}+1}^{L}\frac{\Delta\psi_{\phi,j}}{p_{j}}\right]&\le \mbe\left[\left(\sum_{j={\underline{m}}+1}^{L}\frac{\Delta\psi_{\phi,j}}{p_{j}}\right)^2\right]
\le \sum_{\ell = {\underline{m}}+1}^\infty w_\ell\mbe\left[\left(\sum_{j={\underline{m}}+1}^{\ell}\frac{\Delta\psi_{\phi,j}}{p_{j}}\right)^2\right]\\&\le \sum_{\ell = {\underline{m}}+1}^\infty (\ell-\underline{m})w_\ell\sum_{j={\underline{m}}+1}^{\ell}p_j^{-2}\mbe\left[\Delta\psi_{\phi,j}^2\right]\\&\le B_{1}\sum_{\ell = {\underline{m}}+1}^\infty (\ell-\underline{m})w_\ell\sum_{j={\underline{m}}+1}^{\ell}2^{(2\alpha-r_{1})j}\\&\le\frac{B_{1}p2^{(2\alpha-r_1)(\underline{m}+1)}}{1-2^{2\alpha-r_{1}}}\sum_{\ell = {\underline{m}}+1}^\infty (\ell-\underline{m})2^{-\alpha\ell}(1-2^{(2\alpha-r_{1})(\ell-\underline{m})})\\&\le\frac{B_{1}p2^{(\alpha-r_{1})\underline{m}}}{1-2^{-2\alpha+r_{1}}}\left(\sum_{\ell=1}^\infty \ell 2^{-(r_{1}-\alpha)\ell}-\sum_{\ell=1}^\infty \ell 2^{-\alpha\ell}\right)\\&\le\frac{B_{1}p2^{(\alpha-r_{1})\underline{m}}}{1-2^{-2\alpha+r_{1}}}\left(\frac{2^{-(r_{1}-\alpha)}}{(1-2^{-(r_{1}-\alpha)})^2}-\frac{2^{-\alpha}}{(1-2^{-\alpha})^2}\right)\\&\le\frac{B_{1}(1-2^{-\alpha})2^{(\alpha-r_{1})(\underline{m}+1)}}{(1-2^{-2\alpha+r_{1}})(1-2^{\alpha-r_{1}})^2},
\end{align*}
where $ B_1$ represents the implied constant in the upper bound of the second moment in \Cref{theorem:var_delta_loss_query}. Note that
\begin{align*}
    \mathrm{Var}\left[\psi_{\phi,M_{\underline{m}}}\right]\leq 2\mathrm{Var}\left[\log g_\phi(x,\theta)\right]+2\mathrm{Var}\left[\log \hat{Z}_{M_{\underline{m}}}(x,\phi)\right],
\end{align*}
and $\mathrm{Var}\left[\log \hat{Z}_{M_{\underline{m}}}(x,\phi)\right]\le A_12^{-\underline{m}}$ for a constant $A_1>0$, as shown in \cite{ryan2003estimating}. Finally,
\begin{align*}
    \mathrm{Var}\left[U_{\mathrm{GRR}}\right] \leq 2\mathrm{Var}\left[\log g_\phi(x,\theta)\right]+ A_12^{-\underline{m}} +  \frac{2B_1(1-2^{-\alpha})2^{(\alpha-r_{1})(\underline{m}+1)}}{(1-2^{-2\alpha+r_{1}})(1-2^{\alpha-r_{1}})^2}.
\end{align*}
We can similarly bound $\mathrm{Var}\left[V_{\mathrm{GRR}}\right]$ by replacing $r_1$ with $r_2$ from \Cref{theorem:var_delta_grad_query}, yielding
\begin{align*}
\mathrm{Var}\left[V_{\mathrm{GRR}}\right] \leq 2\mathrm{Var}\left[\nabla_\phi\log g_\phi(x,\theta)\right] + A_22^{-\underline{m}} +  \frac{2B_2(1-2^{-\alpha})2^{(\alpha-r_{2})(\underline{m}+1)}}{(1-2^{-2\alpha+r_{2}})(1-2^{\alpha-r_{2}})^2} ,
\end{align*}
where $A_2>0,B_2>0$ are constants independently of $\underline{m},\alpha,r_2$.

For the cost, we have the following
\begin{align*}
\mathrm{Cost}_{\mathrm{GRR}} &\propto M_02^{{\underline{m}}}+\sum_{\ell = {\underline{m}}+1}^\infty w_{\ell}\sum_{\ell'={\underline{m}}+1}^{\ell} M_02^{\ell'}\\
& = M_02^{{\underline{m}}}+M_0(1-2^{-\alpha})\sum_{\ell = {\underline{m}}+1}^\infty2^{-\alpha\ell}\left(2^{\ell+1}-2^{{\underline{m}}+1}\right)\\
& = M_0 2^{\underline{m}}+M_0\frac{2^{(1-\alpha)({\underline{m}}+1)}}{1-2^{(1-\alpha)}}.
\end{align*}

As expected, a larger $\underline{m}$ leads to lower variances, but it increases the costs. On the other hand, for a fixed $\underline{m}$, a larger $\alpha$ leads to a lower cost, but it raises the variances.

\subsection{Truncated MLMC for APT}
In practice, a truncated variant of the unbiased method can be preferable, which introduces a small, controlled bias that yields lower variance while also bounding computational cost and peak GPU memory, making the estimators more stable under fixed compute budgets. When treating $\nabla_\phi \mathcal{L}(\phi)$ as the summation of the infinite random variable sequence $\{\Delta \rho_{\phi,\ell}\}_{\ell=0}^\infty$, we notice that $\Delta \rho_{\phi,0} = \rho_{\phi,M_0}$ is just the nested APT estimator with a bias of order $\mco(1/M_0)$. For $\Delta \rho_{\phi,\ell}$ with $\ell \geq 1$, it actually contributes to the reduction of bias while simultaneously increasing the variance, and this additional variance grows with $\ell$. An empirical demonstration of the additional variance for $\Delta \rho_{\phi,\ell}$ is shown in \Cref{fig:loss_demo}. 
\begin{figure}[htbp]
\includegraphics[width=0.9\linewidth]{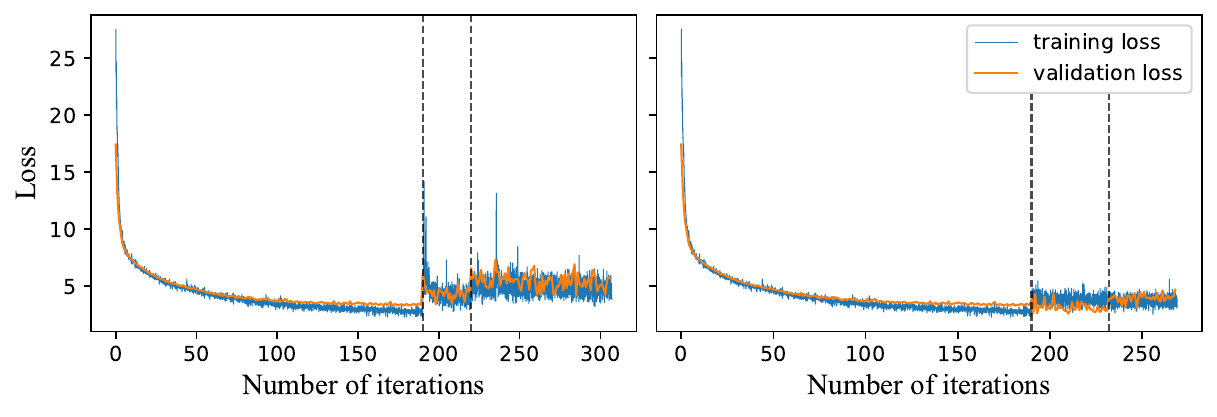}
\caption{\textbf{Left}: Crude RU-MLMC after the third round. \textbf{Right}: The truncated RU-MLMC ($\overline{m}=4$) after the third round.}
\label{fig:loss_demo}
\vspace{-1em}
\end{figure}
We observe that the truncated RU-MLMC method enjoys a more stable loss than the original one, indicating that truncation helps reduce variance. Therefore, to mitigate the variance of the RU-MLMC, we propose to truncate the distribution of $L$ by setting the largest value of $L$ as $\overline{m}$. The truncation leads to a bias of order $\mco(1/M_{{\overline{m}}})$. Following \cite{hu2021bias}, we extend the truncated idea to other unbiased MLMC methods.

\subsubsection{TGRR-MLMC}
TGRR-MLMC stands for the truncated version of GR\\R-MLMC. Let ${\overline{m}}\geq {\underline{m}}$ be a truncated level, and let $L$ be a  random index taking values in $\{{\underline{m}},\dots,{\overline{m}}\}$. The query for the estimation of the gradient of  TGRR-MLMC is then defined as
\begin{align}
    V_{\mathrm{TGRR}} :=\rho_{\phi,M_{{\underline{m}}}}+\sum_{j={\underline{m}}+1}^{L}\frac{\Delta\rho_{\phi,j}}{p_{j}}\label{eq:query_tgrr_grad},
\end{align}
where $p_{j} =\mathbb{P}(L\geq j)>0$ for ${\underline{m}}\le j\le  {\overline{m}}$.
This estimator is biased:
\begin{align*}
    \mathbb{E}[V_{\mathrm{TGRR}}] &= \mathbb{E}[\rho_{\phi,M_{{\underline{m}}}}]+\sum_{\ell={\underline{m}}+1}^{{\overline{m}}} \mathbb{P}(\tilde L=\ell)\sum_{j={\underline{m}}+1}^{\ell}\frac{\mathbb{E}[\Delta \rho_{\phi,j}]}{p_{j}}\\&= \mathbb{E}[\rho_{\phi,M_{{\underline{m}}}}]+\sum_{j={\underline{m}}+1}^{{\overline{m}}}\frac{\mathbb{E}[\Delta \rho_{\phi,j}]}{p_{j}}\sum_{\ell\ge j}^{{\overline{m}}} \mathbb{P}(\tilde L=\ell)\\
    &=\mathbb{E}[\rho_{\phi,M_{{\underline{m}}}}]+\sum_{j={\underline{m}}+1}^{{\overline{m}}} \mathbb{E}[\Delta \rho_{\phi,j}]=\mathbb{E}[ \rho_{\phi,M_{{\overline{m}}}}].
\end{align*}
We should note that if ${\overline{m}}={\underline{m}} =m$, TGRR-MLMC degenerates to nested APT with the inner sample size $M_m = M_02^m$. All truncated methods with truncated level $\overline{m}$ share the bias of order $\mathcal{O}(2^{-\overline{m}})$.

For TGRR-MLMC, the probability mass function of the random index $L$ is chosen as
\begin{equation}\label{neww}
    \mathbb{P}(L=\ell) = 
    \begin{cases}    
    1-\frac{\sum_{\ell={\underline{m}}+1}^{{\overline{m}}} w_\ell}{1-(1-p)^{{\overline{m}}+1}},\ &\ell={\underline{m}}\\
    \frac{w_\ell}{1-(1-p)^{{\overline{m}}+1}},\ &{\underline{m}+1}\le\ell\le {\overline{m}}\\
    0,\ &\text{otherwise},
    \end{cases}      
\end{equation}
where $w_\ell = (1-p)^\ell p$ and $p = 1-2^{-\alpha}$. We denote this distribution by $\mathrm{Ge}(p,{\underline{m}},{\overline{m}})$. The whole procedure of the TGRR-MLMC method is summarized in Algorithm~\ref{alg:TGRR}.

\begin{algorithm2e}[htbp]
\SetEndCharOfAlgoLine{}
\caption{Truncated Generalized Russian Roulette (TGRR)}\label{alg:TGRR}
\KwIn{
Prior $p(\theta)$, 
implicit simulator model $p(x|\theta)$, 
conditional density estimator $q_{F(x,\phi)}(\theta)$, 
$\alpha\in(1,\min(r_1,r_2))$, 
$M_0$, $M_\ell = M_02^\ell$, 
truncated level ${\overline{m}}$, 
base level ${\underline{m}}$, 
total round $R$, 
optimizer $\mathrm{opt}(\cdot,\cdot)$, 
batch size $B$, 
outer sample size $N$ 
}
\textbf{Initialization:}
Set iteration step $t \leftarrow 0$ and initial network parameter $\phi_0$, 
initialize outer sample dataset $\mathcal{D}_{\mathrm{out}} \leftarrow \emptyset$, 
set proposal $\tilde{p}_1(\theta) \leftarrow p(\theta)$

\For{\RouLoop}{
Generate $\{(\theta_i,x_i)\}_{i=1}^N$ from $\tilde{p}_k(\theta)p(x|\theta)$

Update dataset $\mathcal{D}_{\mathrm{out}} \leftarrow \{(\theta_i,x_i)\}_{i=1}^{N} \cup \mathcal{D}_{\mathrm{out}}$

\If{$k=1$}{
\Repeat{\nonstop}{
Sample $\{(\tilde \theta_i,\tilde{x}_i)\}_{i=1}^B$ from $\mathcal{D}_{\mathrm{out}}$ without replacement

Update parameter $\phi_{t+1} \leftarrow \mathrm{opt}\left(\phi_t, -\frac{1}{B}\sum_{i=1}^B \nabla_{\phi_t} \log q_{F(\tilde x_i,\phi_t)}(\tilde\theta_i)\right)$

$t \leftarrow t+1$
}
}

\If{$k\ge2$}{
\Repeat{\nonstop}{
Sample $\{(\tilde \theta_i,\tilde{x}_i)\}_{i=1}^B$ from $\mathcal{D}_{\mathrm{out}}$ without replacement

Sample level $\{\ell_i\}_{i=1}^B$ from $\mathrm{Ge}(1-2^{-\alpha},{\underline{m}}, {\overline{m}})$

\For{\BatLoop}{
Sample $\{\theta'_j\}_{j=1}^{M_{{\underline{m}}}}$ from those $\theta$ in $\mathcal{D}_{\mathrm{out}}$ without replacement

Compute $\rho^{(i)}_{\phi_t,M_{\underline{m}}}=-\nabla_{\phi_t}\log g_{\phi_t}(\tilde x_i,\tilde \theta_i)+\nabla_{\phi_t} \log \frac{1}{M_{\underline{m}}}\sum_{j=1}^{M_{\underline{m}}} g_{\phi_t}(\tilde x_i,\theta'_j)$

\For{\Addloop}{
Sample $\{\theta'_j\}_{j=1}^{M_{\ell}}$ from those $\theta$ in $\mathcal{D}_{\mathrm{out}}$ without replacement

Compute $\Delta \rho^{(i)}_{\phi_t,\ell}$ by \eqref{eq:antit_delta_grad_query} with $x=\tilde{x}_i$, $\phi=\phi_t$ and $\{\theta'_j\}_{j=1}^{M_{\ell}}$
}
Compute the gradient $V_{\mathrm{TGRR}}^{(i)}$ using \eqref{eq:query_tgrr_grad}
}

Update parameter $\phi_{t+1} \leftarrow\mathrm{opt}\left(\phi_t, \frac{1}{B}\sum_{i=1}^B V_{\mathrm{TGRR}}^{(i)}\right)$

$t \leftarrow t+1$
}
}
Update proposal $\tilde{p}_{k+1}(\theta) \leftarrow q_{F(x_o,\phi_t)}(\theta)$
}
\end{algorithm2e}

In this section, we develop three MLMC estimators for the nested APT objective: two unbiased (RU-MLMC, GRR-MLMC) and one biased (TGRR-MLMC). RU-MLMC has the lowest overhead but higher variance; GRR-MLMC reduces variance by deterministically computing the first $\underline{m}$ levels before randomized truncation; TGRR-MLMC deterministically truncates the tail to control runtime and GPU memory, introducing a small, controllable bias, retaining a favorable variance–cost behavior. These methods expose clear bias–variance–cost trade-offs, making each preferable under different compute regimes.

\section{Convergence results of SGD}\label{section:conver_sgd}
In this section, we analyze the convergence of the methods proposed in this paper with some existing techniques. For simplicity, we only study the widely used SGD with constant step size $\gamma$ in this paper. Our results should extend to broader cases. The parameter update procedure of SGD is presented as
\begin{equation}\label{eq:sgd}
\phi_{t+1} = \phi_t-\gamma\nabla_\phi\hat{\mathcal{L}}(\phi_t),
\end{equation}
leveraging the information of an estimator of the gradient of the loss function at the current state $\phi_t$. The gradient estimator for TGRR-MLMC is
\begin{align*}
    \nabla_\phi\hat{\mathcal{L}}^{\text{TGRR}}(\phi) :=\frac{1}{N}\sum_{i=1}^N V_{\mathrm{TGRR}}^{(i)},
\end{align*}
where $V_{\mathrm{TGRR}}^{(i)}$ are iid copies of $V_{\mathrm{TGRR}}$. The gradient estimators for other MLMC methods follow a similar form. The gradient estimator for nested APT is stated previously in \eqref{eq:emp_bias_apt_grad}. Since some of them are biased estimators for $\nabla_\phi \mathcal{L}(\phi)$, we take the following decomposition
\begin{equation}
\nabla_\phi\hat{\mathcal{L}}(\phi) = \nabla_\phi \mathcal{L}(\phi)+b(\phi)+\eta(\phi),
\end{equation}
where $b(\phi)=\mathbb{E}\left[\nabla_\phi\hat{\mathcal{L}}(\phi)\right]-\nabla_\phi \mathcal{L}(\phi)$ and $\eta(\phi)=\nabla_\phi\hat{\mathcal{L}}(\phi)-\mathbb{E}\left[\nabla_\phi\hat{\mathcal{L}}(\phi)\right]$ denote the bias and the noise of gradient estimator by $\nabla_\phi \hat{\mathcal{L}}(\phi)$, respectively. 

\begin{ass}\label{ass:upp_bias_noise}
Assume that
    \begin{align*}
    U_b&:=\underset{\phi \in \Phi}{\sup}\left\|b(\phi)\right\|_2^2 <\infty, \quad
    U_\eta:=\underset{\phi \in \Phi}{\sup}\ \mathbb{E}\left[\left\|\eta(\phi)\right\|_2^2\right] <\infty.
\end{align*}
\end{ass}

We focus on studying how $U_b$ and $U_\eta$ would affect the \textit{optimal gap} at any $t>0$, which is defined as
\begin{align}\label{eq:optimal_gap}
    G_t = \mathbb{E}\left[\mathcal{L}(\phi_t)\right]-\mathcal{L}(\phi^*),
\end{align}
where $\phi^* = \underset{\phi \in \Phi}{\arg\min}\, \mathcal{L}(\phi)$. In situations where a fixed number of iterations is exclusively employed for parameter updates, the evaluation of the impact of variance and bias on the upper bound of the optimal gap is undertaken through the application of the theorem presented below, following \cite{ajalloeian2020convergence}.

\begin{theorem}\label{thm:SGD_opt_gap}
    Assume that \Cref{ass:upp_bias_noise} is satisfied and  there exist positive constants $K,\mu$ such that for every $\phi,\phi'\in \Phi$,
\begin{align}
    \mathcal{L}(\phi)&\leq \mathcal{L}(\phi')+ \nabla_\phi \mathcal{L}(\phi)^T(\phi-\phi')+\frac{K}{2}\left\|\phi-\phi'\right\|^2_2,\label{ass1}\\
    \left\|\nabla_\phi \mathcal{L}(\phi)\right\|_2^2&\geq 2\mu(\mathcal{L}(\phi)-\mathcal{L}(\phi^*)).\label{ass2}
\end{align} If the SGD algorithm \eqref{eq:sgd} has a step size $\gamma\leq \min(1/K, 1/\mu)$, for a fixed $T$ iterations, the upper bound of the optimal gap \eqref{eq:optimal_gap} is given by
\begin{align}\label{eq:decomp_opt_gap}
        G_T \leq (1-\gamma\mu)^{T}G_0+ \frac{1}{2\mu}\left(U_b+U_\eta\right).
    \end{align}
If $\gamma = \min \left\{1/K,\epsilon \mu/(KU_{\eta})\right\}<1/\mu$ and $$T = \max\left\{(K/\mu)\ln(2G_0/\epsilon), KU_{\eta}/(\epsilon \mu^2)\ln(2G_0/{\epsilon})\right\},$$
then the optimal gap satisfies $G_T \leq \epsilon+U_b/(2\mu)$ for any $\epsilon>0$.
\end{theorem}

\begin{proof}
This is a special case of Theorem 6 in \cite{ajalloeian2020convergence} by taking $m=M=0$ there.
\end{proof}

Since $G_0$ depends solely on the initialization of $\phi$, our focus lies on the last two terms of \eqref{eq:decomp_opt_gap}, where $U_b$ and $U_\eta$ have the same impact on the optimal gap $G_T$. This implies that unbiased methods with large variance may ultimately achieve the same result as those biased ones at the same iteration step $T$. 
Hence, in cases where variance dominates bias, methods with smaller variance are preferable. 
On the other hand, when a sufficiently large iteration step $T$ is given and a proper learning rate $\gamma$ is chosen, the upper bound of the optimal gap can be smaller than any given $\epsilon>0$ except for the bias, which is examined in the second part of  \cref{thm:SGD_opt_gap}. 
In this case, unbiased methods are more favorable, as they can converge to an arbitrarily small neighborhood of 0, while biased methods can only converge to a neighborhood of the $U_b/(2\mu)$. 
Moreover, for biased methods with $U_b \gg U_\eta$, the optimal gap for $t_1 \gg t_2$ tends to be $G_{t_1} \approx G_{t_2}$, indicating that increasing the iteration step can be futile in attempting to reduce the optimal gap.

\section{Numerical experiments}\label{sec:UBAPT_exp}
In this section, numerical experiments compare the performance of two unbiased MLMC methods and a biased MLMC method for three models. To handle high-dimensional datasets, it is common to use informative low-dimensional summary statistics of the datasets instead. However, this introduces additional bias, as detailed in \cite{fearnhead2012constructing}. In this paper, we do not look into the effects of using summary statistics for SNPE methods. We refer to Fearnhead and Prangle \cite{fearnhead2012constructing} for a semi-automatic method of constructing summary statistics in the context of ABC.

\textbf{A toy example.} Two-Moon model was studied in Greenberg et al. \cite{apt}. For a given parameter $\theta\in \mathbb{R}^2$, the Two-Moon simulator generates observations $x\in \mathbb{R}^2$ via
\begin{align*}
    a \sim \mathcal{U}(-\pi/2,\pi/2)&,\quad r_2 \sim \mathcal{N}(0.1,0.01^2),\\
    p  = (r_2\cos(a)+0.25,r_2\sin(a))&,\quad x = p+ \left(-\frac{|\theta_1+\theta_2|}{\sqrt{2}},\frac{-\theta_1+\theta_2}{\sqrt{2}}\right).
\end{align*}
The intermediate variables $p$ follow a single crescent-shaped distribution, which is then shifted and rotated around the origin based on the parameter values of $\theta$. The absolute value $|\theta_1+\theta_2|$ contributes to the emergence of a second crescent in the posterior distribution.
We choose a uniform prior over the square $[-1,1]^2$ to perform the inference.  The  observation for this task is $x_o=(0,0)$ following \cite{apt}, the ground truth parameters are $\theta^*=(0.2475,0.2475).$

\textbf{Lotka-Volterra model.}
This model describes the continuous time evolution of a population of predators interacting with a population of prey using a stochastic Markov jump process. The model describes that the birth of a predator at a rate $\exp(\theta_1)XY$, resulting in an increase of $X$ by one. The death of a predator at a rate proportional to $\exp(\theta_2)X$, leading to a decrease of $X$ by one. The birth of a prey at a rate proportional to $\exp(\theta_3)Y$, resulting in an increase of $Y$ by one. The consumption of a prey by a predator at a rate proportional to $\exp(\theta_4)XY$, leading to a decrease of $Y$ by one. Following the experimental details outlined in \cite{papamakarios2019sequential}, we initialize the predator and prey populations as $X = 50$ and $Y = 100$, respectively. We conduct simulations of the Lotka-Volterra model using  Gillespie's algorithm \cite{gillespie1977exact} throughout 30 time units. We recorded the populations at intervals of 0.2 time units, resulting in time series data sets, each consisting of 151 values. The resulting summary statistics $S(x)$ are represented as a 9-dimensional vector, which includes the following time series features: the logarithm of the mean of each time series, the logarithm of the variance of each time series, the auto-correlation coefficient of each time series at lags of 0.2 and 0.4 time units, and the cross-correlation coefficient between the two time series. In our experiments, the prior distribution of the parameters is set to $\mathcal{U}(-5, 2)^4$, and we generate the ground truth posterior with SMC-ABC \cite{beaumont2009adaptive}, which is more costly than the methods in this paper. The ground truth parameters are
\begin{equation*}
\theta^* = (\log 0.01,\  \log 0.5, \  \log 1, \  \log 0.01),
\end{equation*}
the observed summary statistics $S(x_o)$ simulated from the model with the ground truth parameters $\theta^*$ are
\begin{equation*}
S(x_o) = (4.6431,\ 4.0170,\ 7.1992,\ 6.6024,\ 0.9765,\ 0.9237,\ 0.9712,\ 0.9078,\ 0.0476),
\end{equation*}
{and the corresponding standard deviation under the ground truth parameters $\theta^*$ (based on 10,000 simulations) is}
\begin{equation*}
s = (0.3294,\ 0.5483,\ 0.6285,\ 0.9639,\ 0.0091,\ 0.0222,\ 0.0107,\ 0.0224,\ 0.1823).
\end{equation*}

\textbf{M/G/1 queue model.} The M/G/1 queue model \cite{shestopaloff2014bayesian} describes a single server's processing of a queue of continuously arriving jobs. Define $I$ as the total number of jobs that need to be processed, and denote by $s_i$ the processing time required for job $i$. Let $v_i$ be the job's arrival time in the queue, and $d_i$ be the job's departure time from the queue. They satisfy the following conditions
\begin{align*}
s_i \sim \mathcal{U} (\theta_1, \theta_1 + \theta_2), \quad v_i - v_{i-1}  \sim \mathrm{Exp}(\theta_3), \quad d_i - d_{i-1}  = s_i + \max(0, v_i - d_{i-1}).
\end{align*}
In our experiments, we set $I = 50$ jobs and the summary statistics $S(x)$ have been selected as the logarithm of 0th, 25th, 50th, 75th, and 100th percentiles of the set of inter-departure times. The prior distribution of the parameters is given by
\begin{align*}
\theta_1 \sim \mathcal{U} (0, 10),\quad \theta_2 \sim \mathcal{U} (0, 10),\quad \theta_3 \sim \mathcal{U} (0, 1/3),
\end{align*}
and our experiments choose ground truth parameters as $\theta^* = (1,\ 4,\ 0.2).$ The observed summary statistic $S(x_o)$ simulated from the model with the ground truth parameters $\theta^*$ is
\begin{equation*}
S(x_o) = (0.0929,\ 0.8333,\ 1.4484,\ 1.9773,\ 3.1510),
\end{equation*}
and the corresponding standard deviation is
\begin{equation*}
s = (0.1049,\ 0.1336,\ 0.1006,\ 0.1893,\ 0.2918).
\end{equation*}

Our numerical experiments are performed on a computer equipped with a single GeForce RTX 2080 Super GPU and an i9-9900K CPU. The training and inference processes of the model are primarily implemented using the \texttt{PyTorch} package in \texttt{Python}.

\begin{figure}[htbp]
\begin{subfigure}[t]{0.03\textwidth}
  \vspace{2pt}
    \textbf{A}
\end{subfigure}
\hspace{0.02\textwidth}
\begin{subfigure}[t]{0.85\textwidth}
  \vspace{0pt}
    \includegraphics[width=\textwidth]{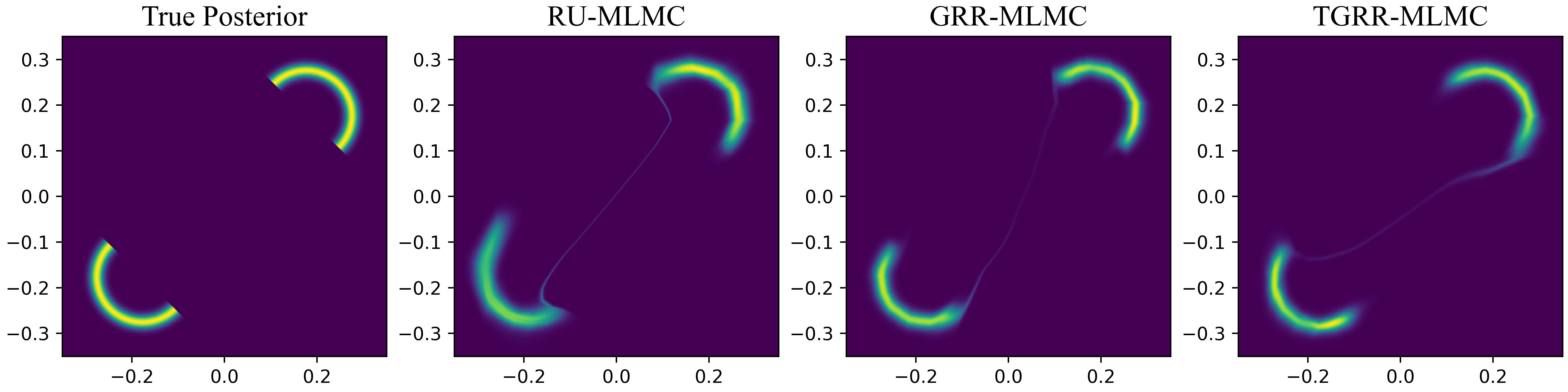}
\end{subfigure}

\begin{subfigure}[t]{0.03\textwidth}
  \vspace{1pt}
    \textbf{B}  
\end{subfigure}
\begin{subfigure}[t]{0.94\textwidth}
  \vspace{0pt}
    \includegraphics[width=\textwidth]{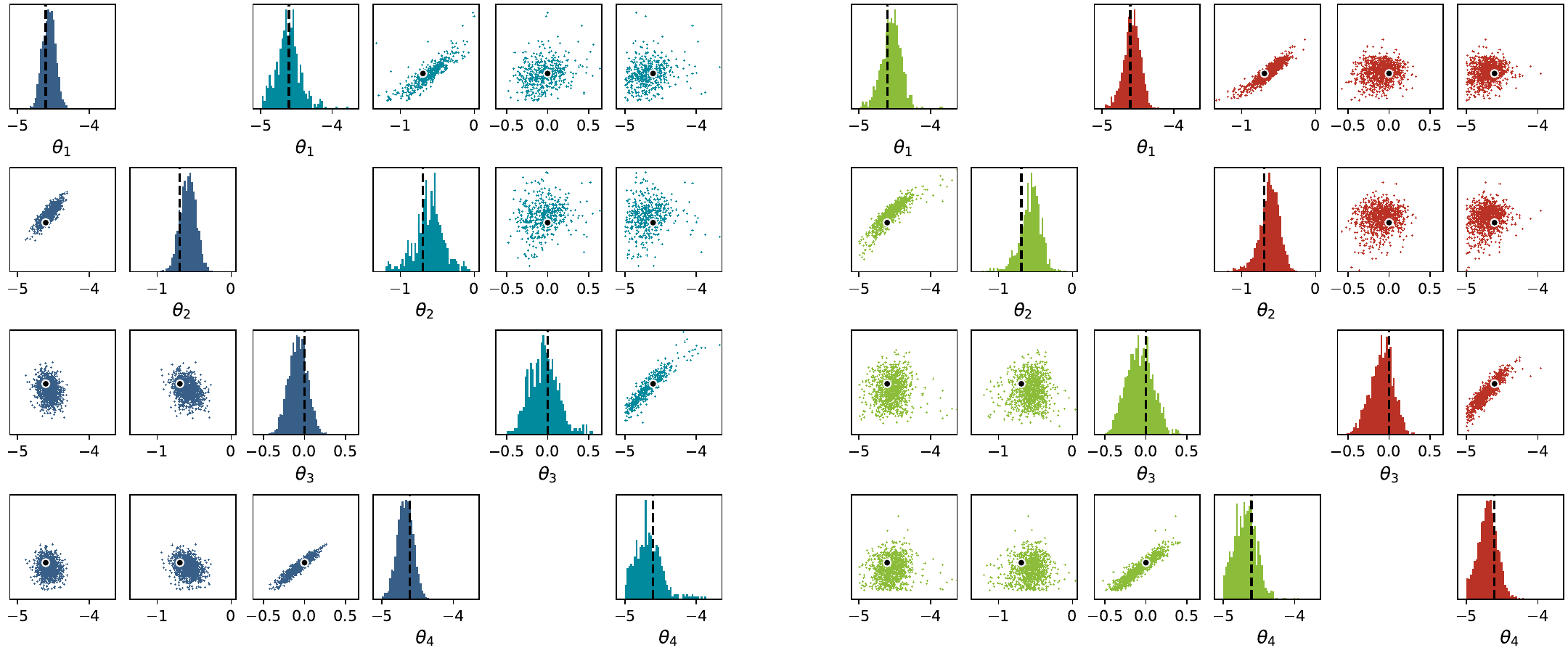}
\end{subfigure}

\begin{subfigure}[t]{0.03\textwidth}
  \vspace{1pt}
    \textbf{C}  
\end{subfigure}
\begin{subfigure}[t]{0.94\textwidth}
  \vspace{0pt}
    \includegraphics[width=\textwidth]{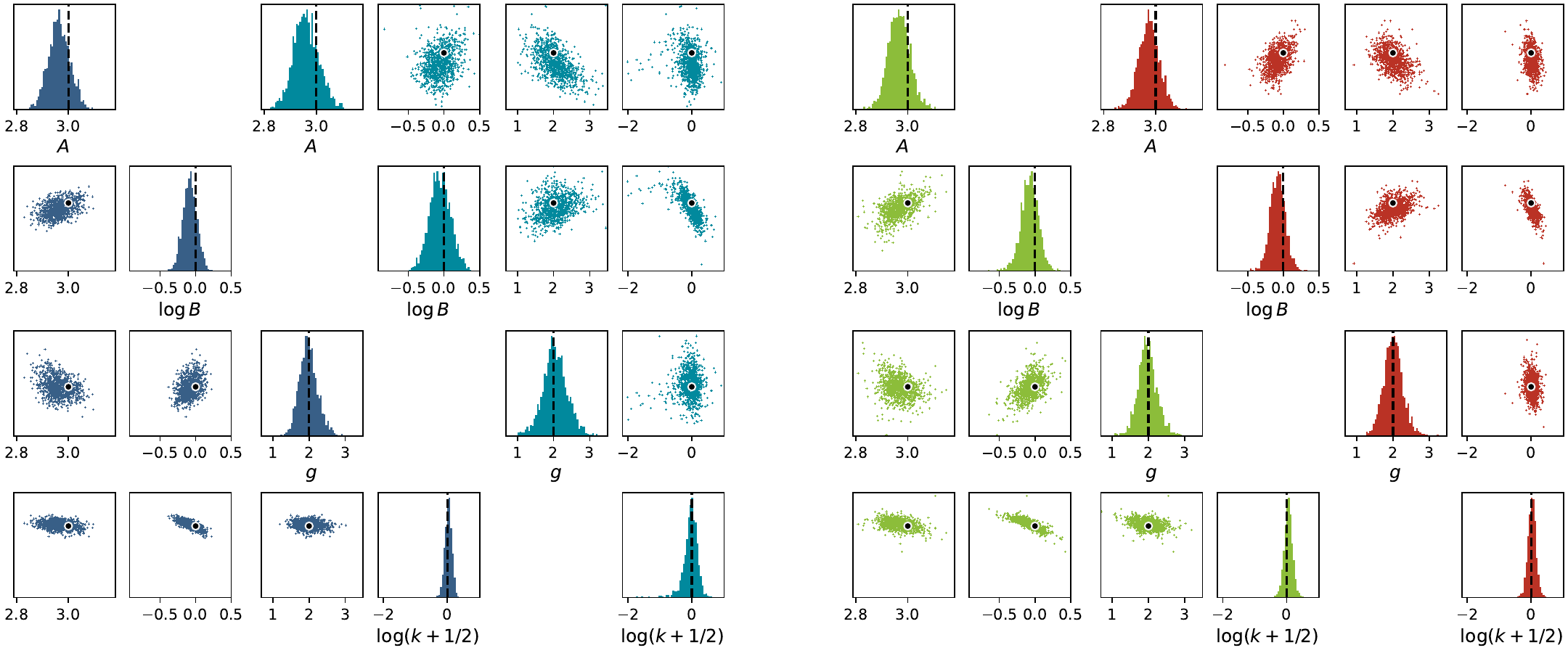}
\end{subfigure}
\caption{\textbf{Approximated posterior for RU-MLMC, GRR-MLMC and TGRR-MLMC} \textbf{A.} Two-Moon, from left to right: available ground truth, RU-MLMC, GRR-MLMC, and TGRR-MLMC. \textbf{B.} Lotka-Volterra, from left to right: ground truth simulated with SMC-ABC \cite{beaumont2009adaptive}, RU-MLMC, GRR-MLMC, and TGRR-MLMC. \textbf{C.} M/G/1 queue model, the setting is the same as Lotka-Volterra. We show a scatter plot for each of the 2D subspaces. The histogram for each $\theta_{i},i=1,\dots,4$  is plotted on the diagonal with the ground truth parameter marked with dotted lines.}
\label{fig:mlmc_compare_density}
\vspace{-1em}
\end{figure}

In the training process, we simulate $N = 1,000$ samples in each round, and with $R = 20$ rounds in total. In each round, we randomly pick 5\% of the newly generated samples $\theta$ and their corresponding $x$ values as validation data. We follow the early stopping criterion proposed by \cite{papamakarios2019sequential}, which terminates the training if the loss value on the validation data does not decrease after 20 epochs in a single round. For the optimizer, we use \texttt{Adam} \cite{kingma2014adam} with a batch size of 100, a learning rate of $1 \times 10^{-4}$, and a weight decay of $1 \times 10^{-4}$. The detailed track of the computational time of various methods can be found in \Cref{appendix:computation_cost}.

In this paper, we employ neural spline flows (NSFs) \cite{durkan2019neural} as the conditional density estimator, which consists of 8 layers. Each layer is constructed using two residual blocks with 50 units and \texttt{ReLU} activation function. With 10 bins in each monotonic piecewise rational-quadratic transform, the tail bound is set to 20. For a detailed introduction to the structure of NSFs, we refer the reader to \Cref{appendix:nsf_intro}.

We compare the performance of RU-MLMC, GRR-MLMC, and TGRR-MLMC on the Two-Moon, Lotka-Volterra, M/G/1 models. To assess the similarity between the approximate posterior distribution $q_{F(x_o, \phi)}(\theta)$ and the true posterior distribution $p(\theta|x_o)$ given the observed data, we employ maximum mean discrepancy (MMD) \cite{apt,gretton2012kernel,hermans2020likelihood,papamakarios2019sequential} and classifier two-sample tests (C2ST) \cite{dalmasso2020validation,gutmann2018likelihood,lopez2016revisiting} as discriminant criteria. Furthermore, we use the logarithmic median distance (LMD) to measure the distance between $x_o$ and $x$ drawn from $p(x|\theta)$, where $\theta$ is sampled from $q_{F(x_o, \phi)}(\theta)$. The negative log-density (NLOG) value of the ground truth parameter $\theta^*$ of the generated observation $x_o$ under the conditional density structure, i.e., $-\log q_{F(x_{o},\phi)}(\theta^*)$, is another useful metric. A lower value of NLOG indicates that, in the parameter space, the posterior density estimator has a higher probability density at the true parameter, which can be used to evaluate the model's accuracy in estimating the true parameter given the observed data \cite{durkan2019neural,apt,hermans2020likelihood,papamakarios2016fast,papamakarios2019sequential}.

For the setting of MLMC methods, we take $M_0 = 8$ and $\overline{m} = 4$ for TGRR-MLMC. The base level for GRR and TGRR is ${\underline{m}}=2$. 
Since the proposal distributions of each round are distinct, the hyperparameter $r_2$ in \Cref{theorem:var_delta_grad_query} differs in each round. To determine the value of $r_2$ for each round, except for the first two, one may take the value from the previous round. To address this issue, we have conducted 50 training processes and sequentially performed linear regression. From the results, we select a global value of $r_2 = 1.8$.

We now discuss the choice of the hyperparameter $\alpha$ for the geometric distribution of the level $L$. Unlike previous work where $\alpha$ is chosen to minimize the average cost $\mathrm{Cost}_{\mathrm{RU}}$ \cite{goda2019multilevelMCest,goda2022unbiased,he2022unbiased}, based on experiment results in \Cref{fig:loss_demo}, in addition to average cost, variance is another crucial factor that demands our attention. Both metrics matter in our evaluation and optimization processes. %
We now try to minimize the \textit{asymptotic inefficiency} \cite{jacob2017unbiased} defined as $H_{\mathrm{RU}} := \mathrm{Var}[V_{\mathrm{RU}}] \times \mathrm{Cost}_{\mathrm{RU}}$.
Utilizing \Cref{theorem:ru_prop_collection}, we find an upper bound for the asymptotic inefficiency
\begin{equation}\label{equation:data_ineff_ru_mlmc}
H_{\mathrm{RU}} \le \frac{B 2^{\alpha + r_2}}{\left(2^{\alpha} - 2\right) \left(2^{r_2} - 2^{\alpha}\right)}
\end{equation}
for some constant $B>0$.
As a result, the optimal $\alpha$ minimizing the upper bound is $\alpha_{\mathrm{RU}}^{*} = (r_2 + 1)/2\in(1, r_2)$. When ${\underline{m}} = 0$ in GRR-MLMC, which degenerates to RR-MLMC, we reach the same conclusion for $\alpha^*_{\mathrm{RR}}$. From this point of view, we can see that RR-MLMC actually trades lower variance for higher average cost compared with RU-MLMC. We apply the same procedure for other MLMC methods, and the corresponding optimal $\alpha$ are presented in \Cref{tab:value_of_opt_alpha}. For an ablation study on the choices of $\alpha$, we refer the readers to \Cref{app:alpha_ablation}.

\begin{table}[t]
\caption{The value of $\alpha$ minimizing the upper bound of asymptotic inefficiency}
    \centering
    \begin{tabular}{c|c|c|c}
    \hline
       \multirow{2}{*}{Method}   & RU-MLMC & GRR-MLMC  &TGRR-MLMC   \\
         &&$({\underline{m}}=2)$&$({\overline{m}}= 4)$\\\hline
         $\alpha\  (1<\alpha<r_2)$&1.4&1.209&1.673\\ \hline 
    \end{tabular}
    \label{tab:value_of_opt_alpha}
\end{table}

\begin{figure}[htbp]
\hspace{0.06\textwidth}
\begin{subfigure}[t]{0.05\textwidth}
  \vspace{3pt}
    \textbf{A}
\end{subfigure}
\begin{subfigure}[t]{0.75\textwidth}
  \vspace{0pt}
    \includegraphics[width=\textwidth]{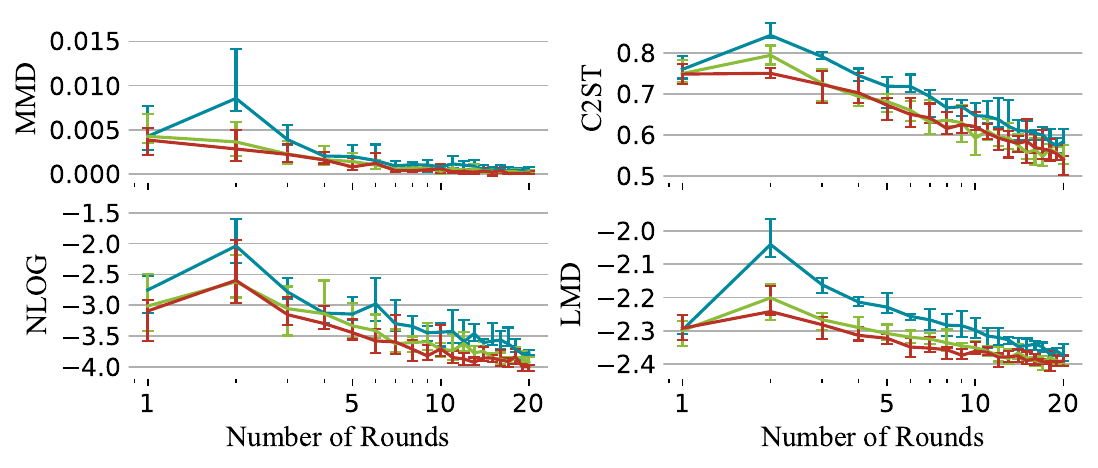}
\end{subfigure}

\hspace{0.06\textwidth}
\begin{subfigure}[t]{0.05\textwidth}
  \vspace{3pt}
    \textbf{B}  
\end{subfigure}
\begin{subfigure}[t]{0.75\textwidth}
  \vspace{0pt}
    \includegraphics[width=\textwidth]{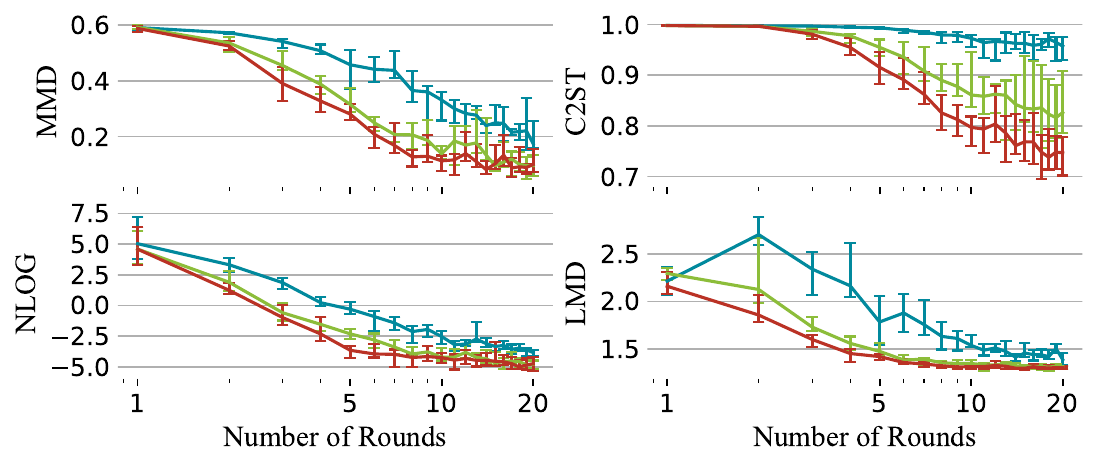}
\end{subfigure}

\hspace{0.06\textwidth}
\begin{subfigure}[t]{0.05\textwidth}
  \vspace{3pt}
    \textbf{C}  
\end{subfigure}
\begin{subfigure}[t]{0.80\textwidth}
  \vspace{3.2cm}
    \includegraphics[width=\textwidth]{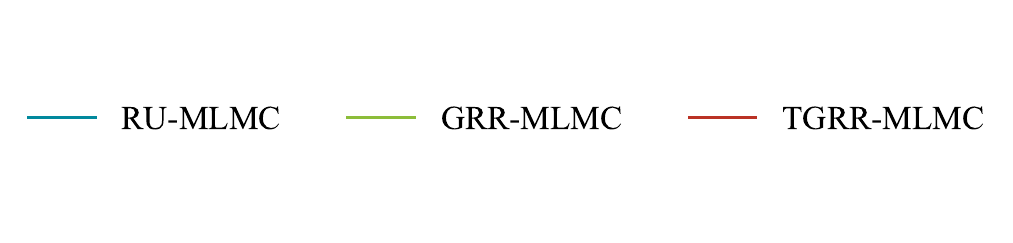}
\end{subfigure}
\hspace{-0.80\textwidth}
\vspace{-1.0cm}
\begin{subfigure}[t]{0.75\textwidth}
  \vspace{0pt}
    \includegraphics[width=\textwidth]{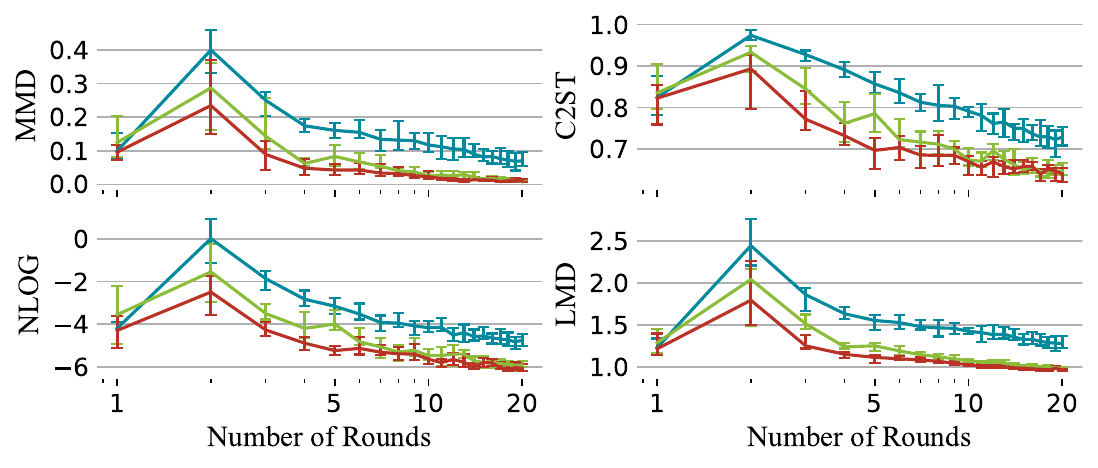}
\end{subfigure}
\caption{\textbf{Performance of RU-MLMC, GRR-MLMC and TGRR-MLMC} \textbf{A.} Two-Moon,  \textbf{B.} Lotka-Volterra \textbf{C.} M/G/1 queue model, blue, green, and red correspond to RU-MLMC, GRR-MLMC, and TGRR-MLMC, respectively.}
\label{fig:mlmc_compare_perf}
\vspace{-1em}
\end{figure}

We present our results for the baseline models in \Cref{fig:mlmc_compare_density} and \Cref{fig:mlmc_compare_perf}. We observe that both unbiased methods are inferior to the biased ones in some cases. When comparing RU-MLMC with GRR-MLMC in the case of the Lotka-Volterra and M/G/1 queue models, we find that the unbiased method can greatly benefit from variance reduction. In the case of Lotka-Volterra, when comparing TGRR and GRR methods, we conclude that the gradient information $\Delta \rho_{\phi,\ell}$ where $\ell \geq \overline{m}$ not only does not contribute to the overall performance but also makes it worse when measured with C2ST. As suggested by \Cref{thm:SGD_opt_gap}, this could be due to the domination of variance, where the effect of reducing variance is superior to reducing bias.

In the case of SNPE, where the complex model of $q_{F(x,\phi)}(\theta)$ is used for density estimation, variance tends to dominate the bias as it is suggested by \Cref{thm:SGD_opt_gap}, indicating that excessive variance could seriously affect the training process of the density estimator. Instead of seeking unbiasedness, one should try to strike a balance between bias, average computational cost, and variance of the gradient in this case.

Taking these empirical trends and the bias-variance trade-off, we conclude the following practical guidance on choosing an appropriate MLMC variant.
\begin{enumerate}
    \item Under tight computational budgets or complex tasks, TGRR-MLMC often yields the most stable training.
    \item When unbiasedness is required and the effect of variance is moderate, GRR-MLMC is a favorable option.
    \item When both unbiasedness and minimal cost are priorities, RU-MLMC is preferable.
\end{enumerate}
In our experiments, this pattern is reflected across tasks. It suffices to use  RU-MLMC for the simple Two-Moon model. GRR-MLMC performs well on the intermediate M/G/1 queue model. TGRR-MLMC is a better choice for the more challenging Lotka–Volterra model.

In \Cref{tab:compare_snse}, we also compare TGRR-MLMC with sequential neural score estimation (SNSE) proposed by \cite{sharrock2022sequential}, which is a state-of-the-art likelihood-free inference method. 
We find that TGRR-MLMC achieves a lower mean C2ST than SNSE on all three tasks; however, it requires 1.2–1.5× wall-clock time, as shown in \Cref{appendix:computation_cost}. Notably,
SNSE consumes up to 60 GB of GPU memory for a batch size of 128, whereas our method requires only 5 GB for a batch size of 500, demonstrating a substantial advantage in memory efficiency. For the hyperparameter of SNSE, we mainly follow the default setting from \cite{durkan2019neural} and change the maximum iteration to 1,000 and samples per round to 1,000 to align with our settings. 

To investigate how the proposed MLMC method scales to higher-dimensional problems, we infer the single-compartment Hodgkin-Huxley type neuron model \cite{pospischil2008minimal,gonccalves2020training} with 8 parameters of interest. The results are deferred to \Cref{app:HH_model}. 

\begin{table}[t]
\centering
\caption{C2ST of different methods after 10,000 simulations. The values represent the mean and standard deviation.}
\label{tab:compare_snse}
\footnotesize
\begin{tabular}{@{}lccccc@{}}
\toprule
Task & Two-Moon & Lotka-Volterra & M/G/1 \\
\midrule
TGRR-MLMC& $0.5407 \pm 0.0315$ & $0.7472 \pm 0.0599$ & $0.6402 \pm 0.0210$ \\
SNSE & $0.5598 \pm 0.0673$ & $0.9153 \pm 0.0280$ & $0.8394 \pm 0.0116$ \\
\bottomrule
\end{tabular}
\end{table}

\begin{table}[tb]
\centering
\caption{Comparison of computational time costs (in minutes). The values represent the mean and standard deviation.}
\label{appendix:computation_cost}
\footnotesize
\begin{tabular}{@{}lccccc@{}}
\toprule
Methods/Tasks & Two-Moon & Lotka-Volterra & M/G/1  \\
\midrule
TGRR-MLMC & $299.35 \pm 28.96$ & $300.10 \pm 26.88$ & $294.65 \pm 18.67$ \\
SNSE &  $194.76 \pm 12.76$ &  $217.64 \pm 18.20$  & $243.80 \pm 17.28$  \\
\bottomrule
\end{tabular}
\end{table}

\section{Concluding remarks}\label{section:discuss}

We revisited APT by replacing the atomic heuristic with a nested formulation to allow rigorous convergence analysis of learning. We explore a family of MLMC nested estimators to accommodate different scenarios, as lower computational complexity alternatives to the single-level nested estimator. Our theory quantifies the orders of bias, variance, and average cost for the loss and gradient estimators. Numerical experiments confirm the effectiveness of the proposed methods for various tasks.

We acknowledge some limitations regarding the proposed methods. Firstly, they may suffer from excessive variance accompanied by high-dimensional problems and complex neural network structures, which lead to inferior performance to some extent. A possible way to reduce the variances of MLMC estimators is to use quasi-Monte Carlo (QMC) methods, which use low-discrepancy (intuitively, more uniform) points to generate samples from the proposal. QMC-based quadrature can yield a faster convergence rate than single-level Monte Carlo, and its application in nested MLMC methods has already been investigated in \cite{he2022unbiased}. Secondly,  MLMC takes more computational time than atomic APT due to the gradient calculation at multiple levels. How to speed up the algorithm of MLMC is left for future work. 

\appendix

\section{Proof for \Cref{theorem:var_delta_grad_query}}\label{appendix:pf_var_delta_grad_query}

We first recall a useful lemma for bounding moments and tail probability of iid sample mean, which was used in  \cite{Giles2017DecisionmakingUU}.

\begin{lemma}\label{lemma:key_lemma1}
Let $X$ be a real-valued random variable with mean zero, and let $\overline{X}_{N}$ be an average of $N$ iid samples of $X$. If for $\mathbb{E}[|X|^u] <\infty$ for $u\geq 2$, there exists a constant $C_u>0$ depending only on $u$ such that
\begin{align*}
\mathbb{E}\left[\left|\overline{X}_{N}\right|^{u}\right]\leq C_{u}\dfrac{\mathbb{E}[|X|^{u}]}{N^{u/2}},\quad \mathbb{P}\left[\left\lvert\overline{X}_N\right\rvert>c\right]\leq C_u\dfrac{\mathbb{E}\left[|X|^{u}\right]}{c^{u}N^{u/2}},
\end{align*}
for any $c>0$.
\end{lemma}
\begin{proof}[Proof of \Cref{theorem:var_delta_grad_query}]
The proof follows \cite{goda2022unbiased}. We first note that
\begin{align}
M_{\max}&:=\underset{x,\phi,\theta}{\sup} \|\nabla_\phi \log g_\phi(x,\theta)\|_{\infty}=\underset{x,\phi,\theta}{\sup} \frac{\|\nabla_\phi  g_\phi(x,\theta)\|_{\infty}}{g_\phi(x,\theta)}<\infty,\label{eq:mmax}\\
H(u)&:=\underset{\phi \in \Phi}{\sup}\ \mathbb{E}\left[\left(\frac{g_\phi(x,\theta)}{Z(x,\phi)}\right)^u\right]<\infty,\ \forall u\in (0,s].\notag
\end{align}

Define the event
\begin{align*}
    \mathcal{A}:=\left\{\left| S^{(a)}_{\ell-1}\right| >\frac{1}{2}\right\}\bigcup\left\{\left|S^{(b)}_{\ell-1}\right| >\frac12\right\},
\end{align*}
where $S^{(i)}_{\ell-1}:= g^{(i)}_{\phi,M_{\ell-1}}(x)/Z(x,\phi)-1$ with $i\in\{a,b\}$, which is a sample mean of $M_{\ell-1}$ iid copies of $ g_\phi(x,\theta)/Z(x,\phi)-1$ with zero mean. We also define
\begin{equation*}
S_{\ell} := \frac{g_{\phi,M_{\ell}}(x)}{Z(x,\phi)}-1=\frac{S^{(a)}_{\ell-1}+S^{(b)}_{\ell-1}}{2}.
\end{equation*}
Let $\mathcal{A}^c$ be the complement of the event $\mathcal{A}$. Thanks to the following decomposition
\begin{equation}\label{eq:2nd_delta_grad_query_dec}
    \mathbb{E}\left[\left\|\Delta \rho_{\phi,\ell}\right\|_2^2\right] = \mathbb{E}\left[\left\|\Delta \rho_{\phi,\ell}\right\|_2^2\mathbf{1}_{\mathcal{A}}\right] +\mathbb{E}\left[\left\|\Delta \rho_{\phi,\ell}\right\|_2^2\mathbf{1}_{\mathcal{A}^c}\right],
\end{equation}
it suffices to develop a proper upper bound for the two terms.

Let $\{a_i\}_{i=1}^d$ be any real finite sequence and $\{b_i\}_{i=1}^d$ be a positive finite sequence, for any $j>0$ we have $a_j \leq  b_j\  \underset{i}{\max}\left\{{a_i}/{b_i}\right\}$. Taking summation on both sides yields
\begin{align}\label{eq:frac_bound2}
   \frac{\sum_{j=1}^d a_j}{\sum_{j=1}^db_j} \leq \underset{i=1,\dots,d}{\max}\left\{\frac{a_i}{b_i}\right\}.
\end{align}
Using \eqref{eq:frac_bound2}, we have
\begin{align*}
    \left\|\nabla_\phi \log \hat{Z}_M(x,\phi)\right\|^2_2 &=\left\|\frac{\sum_{j=1}^{M} \nabla_\phi g_\phi(x,\theta_j')}{\sum_{j=1}^{M} g_\phi(x,\theta_j')}\right\|^2_2=\sum_{i=1}^d\left(\frac{\sum_{j=1}^{M} |\partial_{\phi_i} g_\phi(x,\theta_j')|}{\sum_{j=1}^{M} g_\phi(x,\theta_j')}\right)^2\\& \leq d\left\|\nabla_\phi \log g_\phi(x,\theta)\right\|_{\infty}^2\le dM_{\max}^2,
\end{align*}
where $\phi=(\phi_1,\dots,\phi_d)\in\mathbb{R}^d$.
As a result,
\begin{align*}
    \left\|\rho_{\phi,M}\right\|^2_2 &=\left\|-\nabla_\phi\log g_\phi(x,\theta)+\nabla_\phi \log \hat{Z}_M(x,\phi)\right\|^2_2\\&\leq 2\left\|\nabla_\phi \log g_\phi(x,\theta)\right\|^2_2+2\left\|\nabla_\phi \log \hat{Z}_M(x,\phi)\right\|^2_2\le 4dM_{\max}^2.
\end{align*}
This gives
\begin{align*}
    \max\left\{\left\|\rho^{(a)}_{\phi,M_{\ell-1}}\right\|_2^2,\left\|\rho^{(b)}_{\phi,M_{\ell-1}}\right\|_2^2\right\} \leq 4dM^2_{\max}.
\end{align*}
It then follows 
\begin{align}
\left\|\Delta \rho_{\phi,\ell}\right\|^2_2 &\leq \left(\left\|\rho_{\phi,M_\ell}\right\|_2+\frac{\left\|\rho^{(a)}_{\phi,M_{\ell-1}}\right\|_2}{2}+\frac{\left\|\rho^{(b)}_{\phi,M_{\ell-1}}\right\|_2}{2}\right)^2\notag\\
&\leq 3\left\|\rho_{\phi,M_\ell}\right\|_2^2+\frac{3}{4}\left\|\rho^{(a)}_{\phi,M_{\ell-1}}\right\|_2^2+\frac{3}{4}\left\|\rho^{(b)}_{\phi,M_{\ell-1}}\right\|_2^2 \leq 18dM^2_{\max}.\label{eq:deltarho}
\end{align}
For any $u\in[2,s]$, applying \Cref{lemma:key_lemma1} gives
\begin{align}\label{eq:pA}
\mathbb{P}\left[\mathcal{A}\right]&\leq \mathbb{P}\left[\left| S^{(a)}_{\ell-1}\right| >\frac{1}{2}\right]+\mathbb{P}\left[\left| S^{(b)}_{\ell-1}\right| >\frac{1}{2}\right] \notag
\\&\leq \frac{2C_{u}}{2^{-u}M_{\ell-1}^{u/2}}\mbe[|g_\phi(x,\theta)/Z(x,\phi)-1|^{u}] \notag\\
&\leq \frac{2^{3u/2+1}C_{u}}{M_{\ell}^{u/2}}(\mbe[(g_\phi(x,\theta)/Z(x,\phi))^{u}]+1) \notag\\&\leq \frac{2^{3u/2+1}C_u}{M_{\ell}^{u/2}}(H(u)+1).
\end{align} 
Then for the first term in the right hand side (RHS) of \eqref{eq:2nd_delta_grad_query_dec}, we have 
\begin{align}\label{eq:bound_event_A}
    \mathbb{E}\left[\left\|\Delta \rho_{\phi,\ell}\right\|_2^2\mathbf{1}_{\mathcal{A}}\right]\leq 18dM^2_{\max}\mathbb{P}\left[\mathcal{A}\right]=\mco\left(M_\ell^{-s/2}\right).
\end{align}

For the second term  $\mathbb{E}[\|\Delta \rho_{\phi,\ell}\|_2^2\mathbf{1}_{\mathcal{A}^c}]$ in the RHS of \eqref{eq:2nd_delta_grad_query_dec}, we utilize the antithetic property to obtain the following identity
\begin{align}
    \Delta \rho_{\phi,\ell}
    &=\frac{1}{2}\tilde F_{\ell-1}^{(a)}S_{\ell-1}^{(a)}+\frac{1}{2}\tilde F_{\ell-1}^{(b)}S_{\ell-1}^{(b)}-\tilde F_{\ell}S_{\ell}\notag\\
    &\quad+\frac{1}{2}\frac{\nabla_\phi Z(x,\phi)}{ g^{(a)}_{\phi,M_{\ell-1}}(x)}\left(S_{\ell-1}^{(a)}\right)^2
    +\frac{1}{2}\frac{\nabla_\phi Z(x,\phi)}{ g^{(b)}_{\phi,M_{\ell-1}}(x)}\left(S_{\ell-1}^{(b)}\right)^2-\frac{\nabla_\phi Z(x,\phi)}{ g_{\phi,M_{\ell}}(x)}\left(S_{\ell}\right)^2,\label{eq:dec_delta_grad_query}
\end{align}
where
\begin{align*}
    \tilde{F}^{(i)}_{\ell-1}&:= \left(\nabla_{\phi} g^{(i)}_{\phi,M_{\ell-1}}(x)-\nabla_\phi Z(x,\phi)\right)/g^{(i)}_{\phi,M_{\ell-1}}(x),\ i\in\{a,b\},\\
    \tilde{F}_{\ell}&:= \left(\nabla_{\phi} g_{\phi,M_{\ell}}(x)-\nabla_\phi Z(x,\phi)\right)/g_{\phi,M_{\ell}}(x).
\end{align*}
We denote 
\begin{align*}
F^{(i)}_{\ell-1}&:= \left(\nabla g^{(i)}_{\phi,M_{\ell-1}}(x)-\nabla_\phi Z(x,\phi)\right)/Z(x,\phi),\ i\in\{a,b\},\\
F_{\ell}&:= \left(\nabla g_{\phi,M_{\ell}}(x)-\nabla_\phi Z(x,\phi)\right)/Z(x,\phi),
\end{align*}
which are sample means of the random variable
\begin{equation*}
\mathfrak{f} := \left(\nabla_\phi g_\phi(x,\theta)-\nabla_\phi Z(x,\phi)\right)/Z(x,\phi)
\end{equation*}
with zero mean. Directly applying Jensen's inequality on \eqref{eq:dec_delta_grad_query} gives
\begin{align*}
&\left\| \Delta \rho_{\phi,\ell}\right\|^2_2\leq 2\|\tilde{F}^{(a)}_{\ell-1}\|_2^2(S^{(a)}_{\ell-1})^2+2\|\tilde{F}^{(b)}_{\ell-1}\|_2^2(S^{(b)}_{\ell-1})^2+4\|\tilde{F}_{\ell}\|_2^2(S_{\ell})^2\\
&+ 2\left\|\frac{\nabla_\phi Z(x,\phi)}{g^{(a)}_{\phi,M_{\ell-1}}(x)}\right\|_2^2(S^{(a)}_{\ell-1})^4+2\left\|\frac{\nabla_\phi Z(x,\phi)}{g^{(b)}_{\phi,M_{\ell-1}}(x)}\right\|_2^2(S^{(b)}_{\ell-1})^4+4\left\|\frac{\nabla_\phi Z(x,\phi)}{g_{\phi,M_{\ell}}(x)}\right\|_2^2(S_{\ell})^4.
\end{align*}

Now on event $\mathcal{A}^c$, we have $1/{g^{(a)}_{\phi,M_{\ell-1}}(x)}\leq 2/{Z(x,\phi)}, 1/{g^{(b)}_{\phi,M_{\ell-1}}(x)}\leq 2/{Z(x,\phi)}$. With simple algebra, we have $1/{g_{\phi,M_{\ell}}(x)}\leq 2/{Z(x,\phi)}$. Also we have $|S_\ell| \le 1/2$ on $\mathcal{A}^c$. With the upper bound derived before, on event $\mathcal{A}^c$, we arrive at the following:
\begin{align}
\left\| \Delta \rho_{\phi,\ell}\right\|^2_2 &\leq 8\|F^{(a)}_{\ell-1}\|_2^2(S^{(a)}_{\ell-1})^2+8\|F^{(b)}_{\ell-1}\|_2^2(S^{(b)}_{\ell-1})^2+16\|F_{\ell}\|_2^2(S_{\ell})^2\notag \\
&\quad + 8\left\|\frac{\nabla_\phi Z(x,\phi)}{Z(x,\phi)}\right\|_2^2\left((S^{(a)}_{\ell-1})^4+(S^{(b)}_{\ell-1})^4+2(S_{\ell})^4\right).\label{eq:bdd_norm_delta_grad_query}
\end{align}
Applying Jensen's inequality for any $u>0$ yields
\begin{align}\label{eq:grad_normalize_upp}
    \left\|\frac{\nabla_\phi Z(x,\phi)}{Z(x,\phi)}\right\|_{u}^{u} = \frac{\left\|\mathbb{E} [ \nabla_\phi g_\phi(x,\theta)|x]\right\|_{u}^{u}}{Z(x,\phi)^u}\le\frac{dM_{\max}^u\mathbb{E}_[ g_\phi(x,\theta)|x]^u}{Z(x,\phi)^u}=dM_{\max}^u.
\end{align}
By setting $u=2$, it is sufficient to derive an upper bound for the expectation of the $\|F_{\ell}\|_2^2(S_{\ell})^2$ and $(S_\ell)^4$ in \eqref{eq:bdd_norm_delta_grad_query} on event $\mathcal{A}^c$, since directly setting $\ell$ to $\ell-1$ gives the upper bound for the rest of the terms. By \eqref{eq:mmax} and \eqref{eq:grad_normalize_upp}, we have
\begin{align}\label{eq:var_f_bound}
\mathbb{E}\left[\|\mathfrak{f}\|^{u}_{u}\right] &= \mathbb{E}\left[\left\|\frac{\nabla_\phi g_\phi(x,\theta)-\nabla_\phi Z(x,\phi)}{Z(x,\phi)}\right\|^{u}_{u}\right] \notag\\
&\leq 2^{u-1}\mathbb{E}\left[\left\|\frac{\nabla_\phi g_{\phi}(x,\theta)}{Z(x,\phi)}\right\|_{u}^{u}+\left\|\frac{\nabla_\phi Z(x,\phi)}{Z(x,\phi)}\right\|_{u}^{u}\right] \notag\\
&\leq 2^{u-1}dM_{\max}^u\left(\mathbb{E}\left[\left(\frac{g_{\phi}(x,\theta)}{Z(x,\phi)}\right)^{u}\right]+1\right) \notag\\
&\le d2^{u-1}M_{\max}^u(H(u)+1)
.
\end{align}

If $|x-1|\le 1/2$, we then have for any $\tau\in (0,4]$,
\begin{align}\label{eq:exp_bound}
    |x-1|^\tau \leq 2^{\min(s,4)-4}\left|x-1\right|^{\min(s,4)-(4-\tau)}=2^{2r_2-4}\left|x-1\right|^{2r_2-(4-\tau)},
\end{align}
where $r_2=\min(s,4)/2$. Applying \eqref{eq:exp_bound} with $\tau=2$ and H\"older inequality  obtains
\begin{align}
\mathbb{E}\left[\left\|F_\ell\right\|_2^2\left(S_\ell\right)^2\textbf{1}_{\mathcal{A}^c}\right]&\leq \mathbb{E}\left[\left\|F_\ell\right\|_2^2 \times2^{2r_2-4}\left|S_\ell\right|^{{2r_2-2}}\right]\notag \\
&\leq 2^{2r_2-4}\left(\mathbb{E}\left[\left\|F_\ell\right\|_2^{2r_2}\right]\right)^{1/r_2} \left(\mathbb{E}\left[\left|S_\ell\right|^{{2r_2}}\right]\right)^{1-1/r_2}.\label{eq:3rd_term_grad_query}
\end{align}
To obtain the desired upper bound, we handle these two moments separately. Since $\mathbb{E}\left[\|\mathfrak{f}\|^{u}_{u}\right]$ is bounded by \eqref{eq:var_f_bound}, applying Jensen's inequality with \Cref{lemma:key_lemma1} gives
\begin{equation}\label{eq:fell}
\mathbb{E}\left[\left\|F_\ell\right\|_2^{u}\right] \leq \frac{d^{u/2-1}C_{u}}{M_\ell^{u/2}}\mathbb{E}\left[\|\mathfrak{f}\|^{u}_{u}\right]\le\frac{d^{u/2}2^{u-1}C_{u}M_{\max}^{u}(H(u)+1)}{M_\ell^{u/2}}.
\end{equation}
We directly apply \Cref{lemma:key_lemma1} to have
\begin{equation}\label{eq:sell}
\mathbb{E}\left[\left|S_\ell\right|^{{u}}\right]\leq \frac{C_{u}}{M_\ell^{u/2}}\mbe[|g_\phi(x,\theta)/Z(x,\phi)-1|^{u}]\le \frac{C_{u}}{M_\ell^{u/2}}(H(u)+1).
\end{equation}
Using \eqref{eq:3rd_term_grad_query}, \eqref{eq:fell} and \eqref{eq:sell} gives
\begin{align*}
\mathbb{E}\left[\left\|F_\ell\right\|_2^2\left(S_\ell\right)^2\textbf{1}_{\mathcal{A}^c}\right]=\mco(M_\ell^{-r_2}).
\end{align*}
With similar treatment for $\mathbb{E}[\|F^{(a)}_{\ell-1}\|_2^2(S^{(a)}_{\ell-1})^2 \mathbf{1}_{\mathcal{A}^c}]$ and $\mathbb{E}[\|F^{(b)}_{\ell-1}\|_2^2(S^{(b)}_{\ell-1})^2 \mathbf{1}_{\mathcal{A}^c}]$, we then have an upper bound of $\mco(M_\ell^{-r_2})$ for them. 

Applying \eqref{eq:exp_bound} with $\tau=4$, we have
\begin{equation*}
\mathbb{E}[(S_{\ell})^4 \mathbf{1}_{\mathcal{A}^c}]\leq 2^{2r_2-4}\mathbb{E}\left[\left|S_\ell\right|^{2r_2}\right]\leq 2^{2r_2-4}\frac{C_{2r_2}}{M_\ell^{r_2}}(H(2r_2)+1).
\end{equation*}
Similarly, $\mathbb{E}[(S_{\ell-1}^{(a)})^4 \mathbf{1}_{\mathcal{A}^c}]$ and $\mathbb{E}[(S_{\ell-1}^{(b)})^4 \mathbf{1}_{\mathcal{A}^c}]$ have an upper 
 of bound $\mco(M_\ell^{-r_2})$.
Combining these upper bounds, we arrive at
\begin{equation*}
\mathbb{E}\left[\left\|\Delta \rho_{\phi,\ell}\right\|_2^2\right] = \mco(M_\ell^{-s/2})+\mco(M_\ell^{-r_2})=\mco(M_\ell^{-r_2}).
\end{equation*}

The same decomposition on event $\mathcal{A}$ is also studied for $\mathbb{E}\left[\left\|\Delta \rho_{\phi,\ell}\right\|_2\right]$, yielding
\begin{align}
\mathbb{E}\left[\left\|\Delta \rho_{\phi,\ell}\right\|_2\right] = \mathbb{E}\left[\left\|\Delta \rho_{\phi,\ell}\right\|_2\mathbf{1}_{\mathcal{A}}\right] +\mathbb{E}\left[\left\|\Delta \rho_{\phi,\ell}\right\|_2\mathbf{1}_{\mathcal{A}^c}\right].\label{eq:deltarho1}
\end{align}
Using \eqref{eq:deltarho} and \eqref{eq:pA} with $u=2$, we have
\begin{equation*}
\mathbb{E}\left[\left\|\Delta \rho_{\phi,\ell}\right\|_2\mathbf{1}_{\mathcal{A}}\right]=\mco(M_\ell^{-1}).
\end{equation*}
It follows from \eqref{eq:dec_delta_grad_query} that on event $\mathcal{A}^c$, 
\begin{align}
\left\|\Delta \rho_{\phi,\ell}\right\|_2
        &\leq \|F^{(a)}_{\ell-1}\|_2|S^{(a)}_{\ell-1}|+\|F^{(b)}_{\ell-1}\|_2|S^{(b)}_{\ell-1}|+2\|F_{\ell}\|_2|S_\ell|\notag\\
 &+ \left\|\frac{\nabla_\phi Z(x,\phi)}{Z(x,\phi)}\right\|_2\left((S^{(a)}_{\ell-1})^2+(S^{(b)}_{\ell-1})^2+2(S_\ell)^2\right).
\end{align}
For the upper bound of $\mathbb{E}\left[\left\|\Delta \rho_{\phi,\ell}\right\|_2\mathbf{1}_{\mathcal{A}^c}\right]$, we can similarly apply \Cref{lemma:key_lemma1} to derive an upper bound of $\mco(M_\ell^{-1})$ for $\mathbb{E}\left[(S^{(a)}_{\ell-1})^2\mathbf{1}_{\mathcal{A}^c}\right]$, $\mathbb{E}\left[(S^{(b)}_{\ell-1})^2\mathbf{1}_{\mathcal{A}^c}\right]$ and $\mathbb{E}\left[(S_{\ell})^2\mathbf{1}_{\mathcal{A}^c}\right]$. Applying H\"{o}lder's inequality and using \eqref{eq:fell} and \eqref{eq:sell} with $u=2$ yield 
\begin{equation*}
\mathbb{E}\left[\left\|F_\ell\right\|_{2}\left|S_\ell\right| \mathbf{1}_{\mathcal{A}^{c}}\right] 
\leq \left(\mathbb{E}\left[\left\|F_\ell\right\|_2^2\right]\right)^{1/2}\left(\mathbb{E}\left[S_\ell^2\right]\right)^{1/2}=\mco(M_\ell^{-1}).
\end{equation*}
Finally, we have $\mathbb{E}\left[\left\|\Delta \rho_{\phi,\ell}\right\|_2\right]=\mco(M_\ell^{-1}).$
\end{proof}

\section{Normalizing flows as a density estimator}\label{appendix:nsf_intro}

Normalizing flow model maps noise $\mathbf{x}\sim \pi$ to data $\mathbf{y}$ through an invertible and differentiable transformation $f$, that is $\mathbf{y} = f(\mathbf{x})$. Then the probability density of $\mathbf{y}$ under the flow is obtained with a change of variables:
\begin{equation*}
p(\mathbf{y})=\pi(f^{-1}(\mathbf{y}))\left|\mathrm{\det}\left(  \frac{\partial f^{-1}}{\partial \mathbf{y}} \right) \right|.
\end{equation*}
Therefore, it is determined by the bijection $f$ and its inverse $f^{-1}$, which are often implemented by composing a series of invertible neural-network modules. 
Neural spline flows~\cite{durkan2019neural} used in our paper is built upon coupling layers, it %
maps $\mathbf{x}$ to $\mathbf{y}$ by using monotonic rational-quadratic splines, where each interval is defined by a monotonic rational-quadratic function taking the form of a ratio of two quadratic polynomials. These functions are easy to differentiate and, when restricted to monotonic segments, they are also analytically invertible. Moreover, both the derivative and the height at each knot are strictly more flexible. The spline itself maps the interval $[-B, B]$ to $[-B, B]$, and outside this range, the transformation is defined as the identity, resulting in a linear `tail' that allows the overall transformation to accept unconstrained inputs.

The spline uses $K$ different rational-quadratic functions, with boundaries set by $K+1$ spline knots $\{u^{(k)}, v^{(k)}\}_{k=0}^{K}$, with boundaries given by the coordinates $(u^{(0)}, v^{(0)}) = (-B, -B)$ and $(u^{(K)}, v^{(K)}) = (B, B)$. The remaining knots must satisfy the monotonicity conditions 
\begin{equation*}
u^{(k)} < u^{(k+1)}, \quad v^{(k)} < v^{(k+1)}.
\end{equation*}
We specify the spline using $K-1$ arbitrary positive derivative values $\{\delta^{(k)}\}_{k=1}^{K-1}$ at the $K-1$ interior knots $(u^{(k)}, v^{(k)})_{k=1}^{K-1}$, and set the boundary derivatives to $\delta^{(0)} = \delta^{(K)} = 1$ in order to meet the linear tails. 

The coupling transformation $\Psi$ maps the input $\mathbf{x}\in \mathbb{R}^D$ to the output $\mathbf{y}\in \mathbb{R}^D$ as follows:
\begin{enumerate}
    \item Split the input $\mathbf{x}$ into two parts, $\mathbf{x} = [\mathbf{x}_{1:D'-1}, \mathbf{x}_{D':D}]$, where $D'<D$.
    \item A neural network is used to take $\mathbf{x}_{1:D'-1}$ as input and output an unconstrained parameter vector $\phi_i$ of length $3K - 1$ for each $i = D', \ldots, D$.
    \item Compute $\mathbf{y}_i = \Gamma_{\phi_i}(\mathbf{x}_i)$ in parallel for $i = D', \dots, D$, where $\Gamma_{\phi_i}$ is an invertible function parameterized by $\phi_i = [\phi_i^w, \phi_i^h, \phi_i^d]$, where $\phi_i^w$ and $\phi_i^h$ are of length $K$, and $\phi_i^d$ is of length $K - 1$, vectors $\phi_i^w$ and $\phi_i^h$ are used to obtain the widths and heights of $K$ intervals. The cumulative sums of these widths and heights define the $K + 1$ knot points $\{(u^{(k)}, v^{(k)})\}_{k=0}^{K}$. Vector $\phi_i^d$ is obtained the derivative values $\{\delta^{(k)}\}_{k=1}^{K-1}$ at the interior knots.
    \item Set $\mathbf{y}_{1:D'-1} = \mathbf{x}_{1:D'-1}$ and return $\mathbf{y} = [\mathbf{y}_{1:D'-1}, \mathbf{y}_{D':D}]$.
\end{enumerate}

For the monotonic rational-quadratic function in the $k^{\mathrm{th}}$ bin, $w^{(k)} = u^{(k+1)} - u^{(k)}$ denotes the width of the interval, and  $s^{(k)} = (v^{(k+1)} - v^{(k)}) / w^{(k)}$ denotes the slope of the line connecting the coordinates. For $u \in [u^{(k)}, u^{(k+1)}]$,   the monotone invertible mapping is defined by
\begin{equation*}
    \Gamma_{\phi_i}(u) = v^{(k)} + \frac{(v^{(k+1)} - v^{(k)}) [s^{(k)} \xi^2 + \delta^{(k)} \xi (1 - \xi)]}{s^{(k)} + [\delta^{(k+1)} + \delta^{(k)} - 2 s^{(k)}] \xi (1 - \xi)},
\end{equation*}
where  $\xi = (u - u^{(k)}) / w^{(k)}\in [0, 1]$, $k = 0, \dots, K-1$.

\begin{figure}[htbp]
\begin{subfigure}[t]{0.05\textwidth}
  \vspace{2pt}
    \textbf{A}
\end{subfigure}
\hspace{0.025\textwidth}
\begin{subfigure}[t]{0.80\textwidth}
  \vspace{0pt}
    \includegraphics[width=\textwidth]{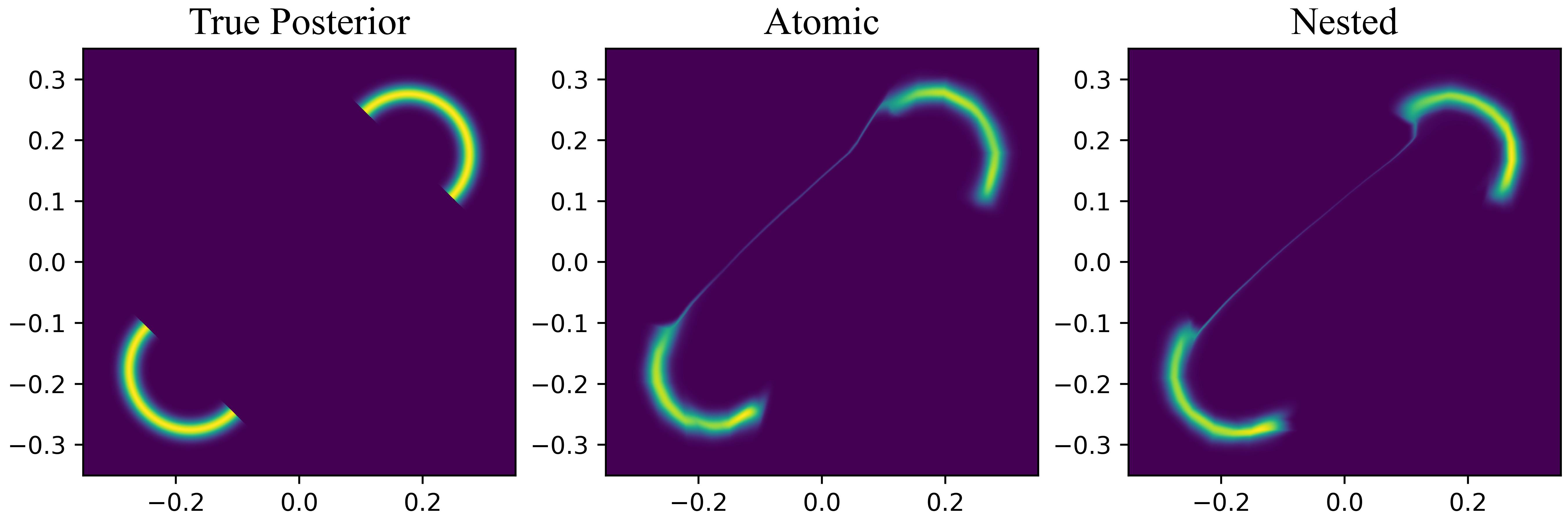}
\end{subfigure}

\begin{subfigure}[t]{0.05\textwidth}
  \vspace{1pt}
    \textbf{B}  
\end{subfigure}
\begin{subfigure}[t]{0.90\textwidth}
  \vspace{0pt}
    \includegraphics[width=\textwidth]{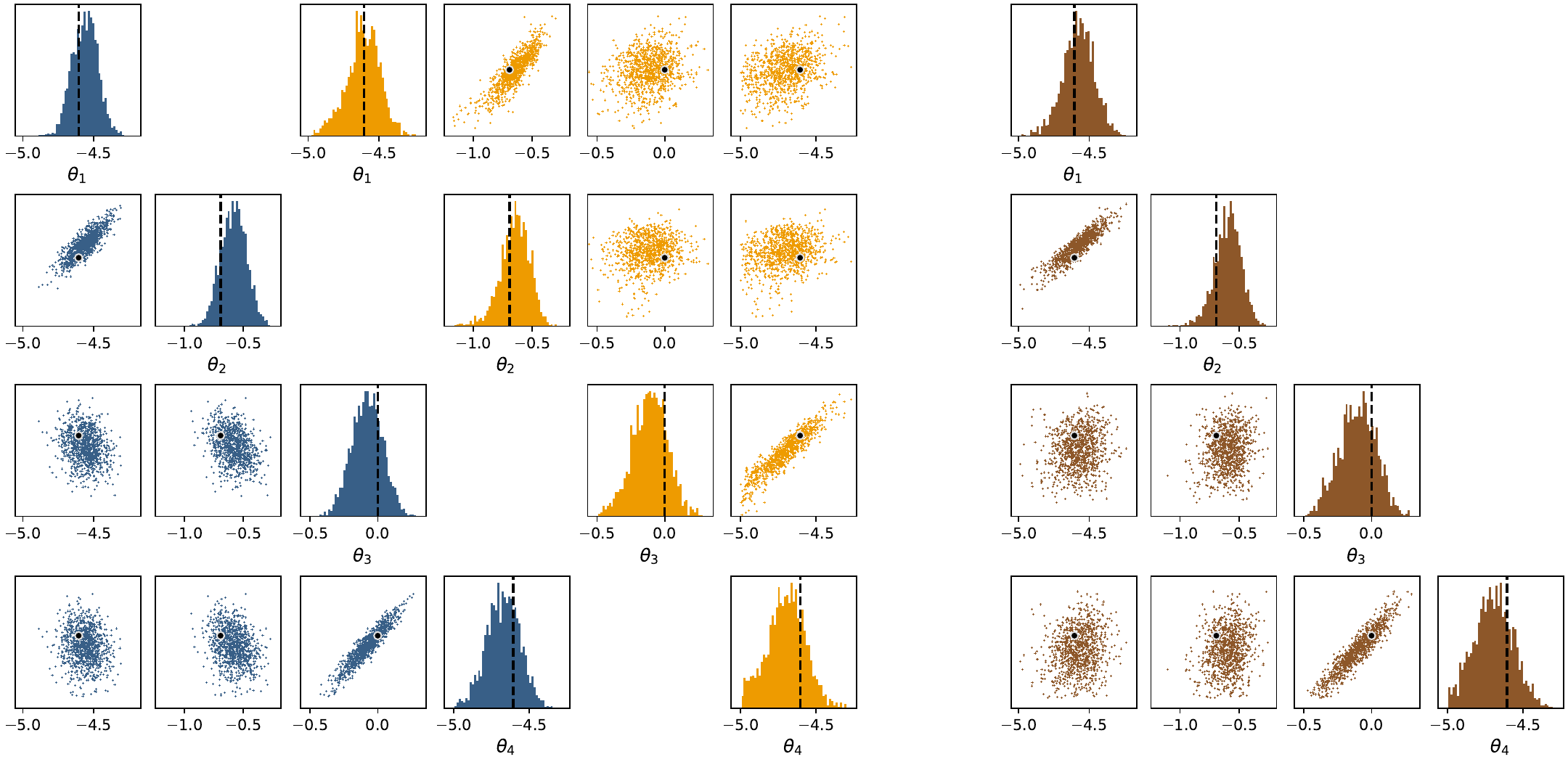}
\end{subfigure}

\begin{subfigure}[t]{0.05\textwidth}
  \vspace{1pt}
    \textbf{C}  
\end{subfigure}
\begin{subfigure}[t]{0.90\textwidth}
  \vspace{0pt}
    \includegraphics[width=\textwidth]{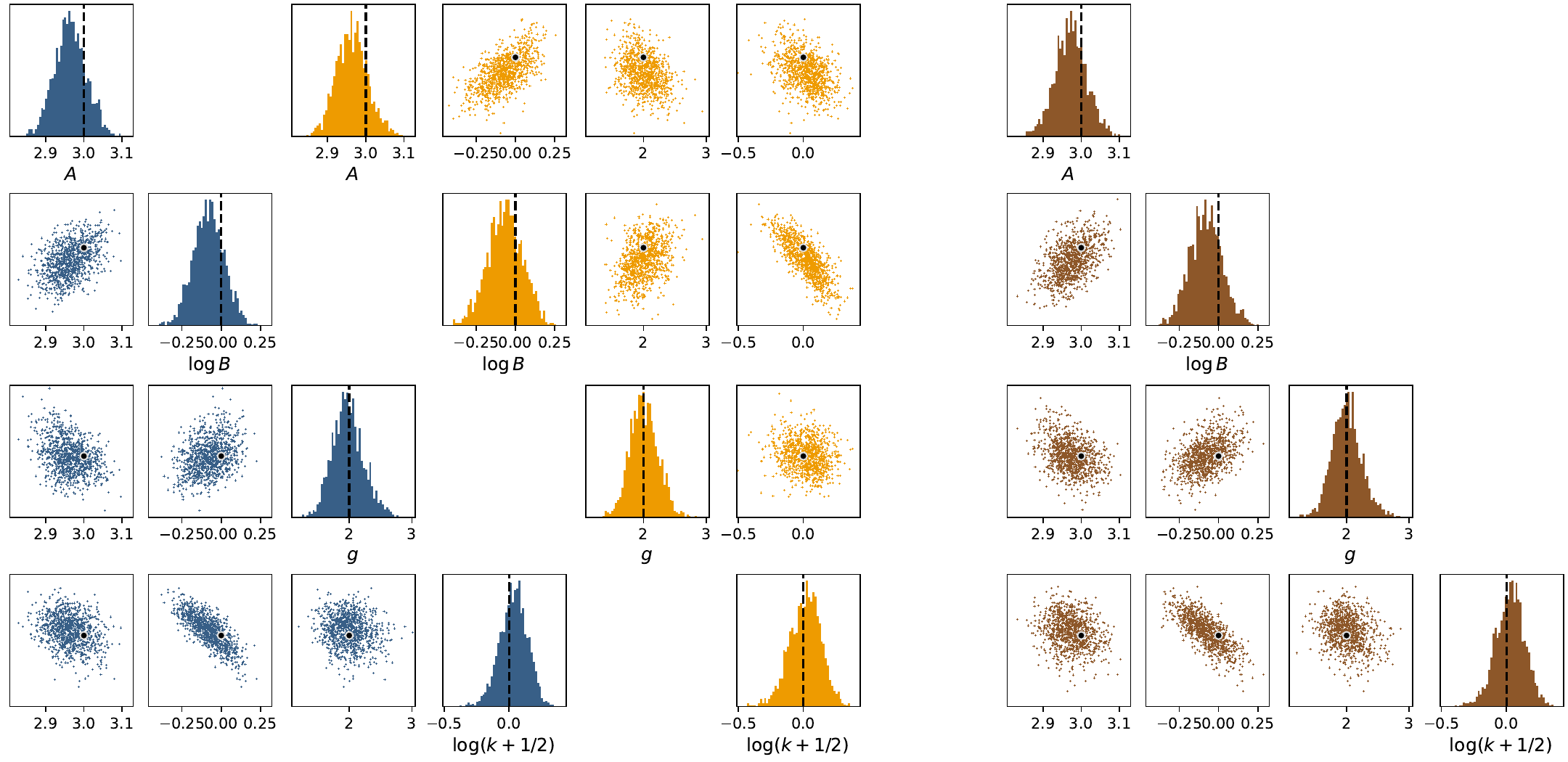}
\end{subfigure}
\caption{\textbf{Approximated posterior for nested APT and atomic APT.} \textbf{A.} Two-Moon model, from left to right: available ground truth, atomic APT with inner samples $M = 100$, nested APT with inner samples $M=100$. \textbf{B.} Lotka-Volterra model, from left to right: ground truth simulated with SMC-ABC \cite{beaumont2009adaptive}, atomic APT with $M = 100$, nested APT with $M=100$. \textbf{C.} M/G/1 queue model, the settings are the same as Lotka-Volterra. We have plotted the scatter plot for each of the 2D subspaces. The histogram for each $\theta_{i},i=1,\dots,4$ is plotted on the diagonal with the ground truth parameter marked with dotted lines.}
\label{fig:nested_compare}
\vspace{-1.5em}
\end{figure}

\section{Visualization of nested APT and atomic APT}\label{appendix:exp_nested_atomic_APT}

Under the settings in \Cref{sec:UBAPT_exp}, \Cref{fig:nested_compare} compares the approximated posteriors for nested APT and atomic APT. Both methods yield a comparable posterior estimation to the benchmark SMC-ABC. A noteworthy advantage of the proposed nested APT method is that it enables the application of a series of existing results for an optimizer with a biased gradient.

\section{Ablation studies for $\alpha$}\label{app:alpha_ablation}

\Cref{appendix:alpha_ablation_table} shows ablation studies of $\alpha$ for TGRR-MLMC, where $\alpha=1.673$ is our proposed optimal value, $\alpha=1.8$ is the value that minimizes the expected cost. We observe that $\alpha=1.673$ exhibits a slightly better performance than $\alpha=1.8$ for the Lotka-Volterra model.

\begin{figure}[htbp]
    \centering
    \includegraphics[width=1.0\linewidth]{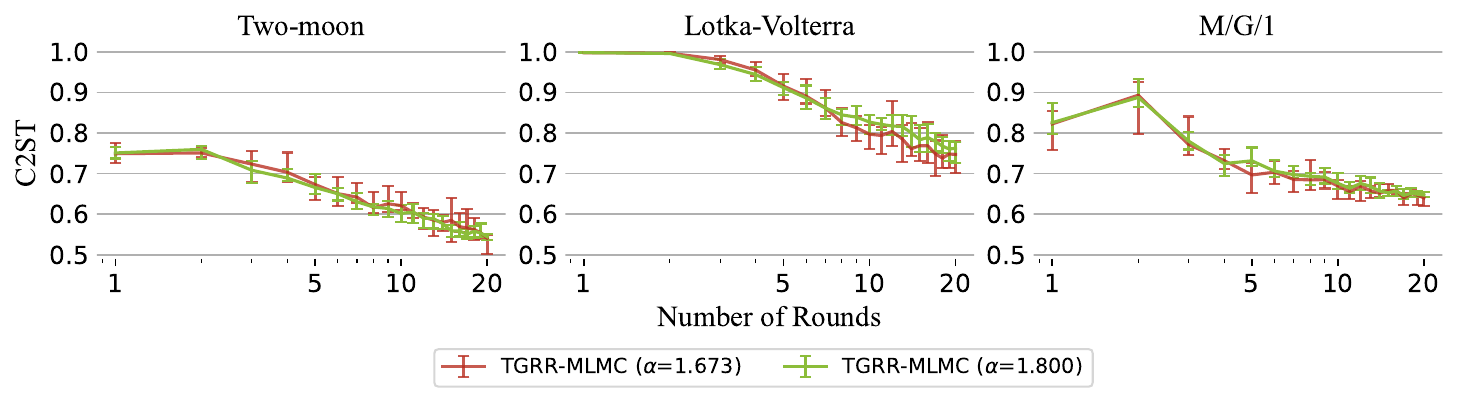}
    \caption{Ablation studies of $\alpha$ for TGRR-MLMC, where $\alpha=1.673$ is our proposed optimal value, $\alpha=1.8$ is the value that minimizes the expected cost.}
    \label{appendix:alpha_ablation_table}
    \vspace{-1em}
\end{figure}

\section{Inference on high-dimensional neuron models}\label{app:HH_model}
We modeled a single-compartment Hodgkin-Huxley type neuron based on the equations in \cite{pospischil2008minimal}, where the dynamics of the membrane potential $V$ and channel gating variables $q \in \{m, h, n, p\}$ are governed by:
\begin{align*}
C_m \frac{dV}{dt} &= g_l (E_L - V) + g_{Na} m^3 h (E_{Na} - V) + g_K n^4 (E_K - V) \\ & \quad + g_M p (E_K - V) + I_{\text{inj}} + \sigma \eta(t), \\
\frac{dq}{dt} &= q_\infty(V) - q / \tau_q(V), \quad q \in \{m, h, n, p\},
\end{align*}
where $C_m$ is the membrane capacitance, $g_L$ is the leak conductance, $E_L$ is the leak reversal potential, and $g_c$ is the density of channels for ion types $c \in \{Na^+, K^+, M\}$. The gating variables $m, h, n, p$ govern the activation and inactivation of the respective channels, and their kinetics are dependent on the membrane potential $V$. The injected current $I_{\text{inj}}$ and Gaussian noise $\sigma \eta(t)$ drive the neuron dynamics. The steady-state functions $q_\infty(V)$ and the time constants $\tau_q(V)$ characterize the gating kinetics, with additional parameters controlling the spike threshold $V_T$ and the time constant of adaptation $\tau_p(V)$.

Following \cite{gonccalves2020training}, we set $E_{Na} = 53 \, \text{mV}$ and $E_K = -107 \, \text{mV}$. The parameters we aim to infer include: the maximal conductances $g_{Na}, g_K, g_l, g_M$, the membrane threshold $V_T$, the leak reversal potential $E_l$, the maximum time constant $t_{\max}$, and the noise parameter $\text{noise}$, corresponding to the vector $\theta\in\mathbb{R}^8$ in the code implementation. 

\begin{figure}[htbp]
    \centering
    \includegraphics[width=0.8\linewidth]{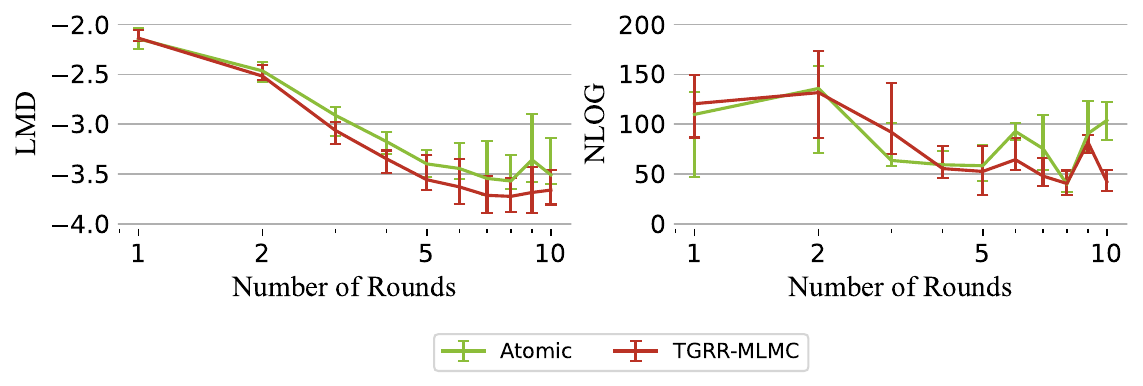}
    \caption{The comparison between atomic APT and TGRR-MLMC on Hodgkin-Huxley, measured with LMD and NLOG since ground truth posterior samples are inaccessible.}
    \label{appendix:hh_model}
    \vspace{-1em}
\end{figure}

\Cref{appendix:hh_model} presents the LMD and NLOG values of the posterior distributions derived from the atomic APT and TGRR-MLMC methods over 10 rounds (10,000 simulations per round). In this high-dimensional task, we observe that the MLMC-based approach is able to improve performance, while atomic APT yields larger LMDs. However, we also find that  MLMC  requires more computational time due to the gradient calculation at multiple levels. The computational cost measured in minutes is $205.77 \pm 18.38$ for atomic APT, and $379.10 \pm 48.71$ for TGRR-MLMC. 

\bibliography{refs}

\end{document}